\documentclass[11pt]{article}

\usepackage{amsfonts, amsmath, amssymb, amsthm}
\usepackage{bbm}
\usepackage{versions}
\usepackage[small]{complexity}
\usepackage{stmaryrd}
\usepackage{layout}
\usepackage[margin=1in,nohead]{geometry}
\usepackage[pdftex,pagebackref=true]{hyperref} % extensive support for hypertext in LaTeX
\def\myeta{\eta}
\def\mygamma{\gamma}

\def\({\left(}
\def\){\right)}

\DeclareMathOperator{\sgn}{sgn}

\renewcommand{\C}{\mathbb{C}}

\newcommand{\B}{\{0,1\}}
\newcommand{\ind}[1]{\mathbbm{1}_{#1}}
\newcommand{\dist}{\text{dist}}

\begin{document}

%\hbadness=10000
%\vbadness=10000

%\setlength{\parskip}{\medskipamount}

%\setlength{\parindent}{0in}

\newtheorem{THM}{Theorem}
\newtheorem*{theorem*}{Theorem}
\newtheorem{theorem}{Theorem}[section]
\newtheorem{corollary}[theorem]{Corollary}
\newtheorem{conjecture}[theorem]{Conjecture}
\newtheorem{lemma}[theorem]{Lemma}
\newtheorem*{lemma*}{Lemma}
\newtheorem{observation}[theorem]{Observation}
\newtheorem{construction}[theorem]{Construction}
\newtheorem{proposition}[theorem]{Proposition}
\newtheorem{definition}[theorem]{Definition}
\newtheorem{claim}[theorem]{Claim}
\newtheorem{fact}[theorem]{Fact}
\newtheorem{assumption}[theorem]{Assumption}
\newtheorem{notation}[theorem]{Notation}
\theoremstyle{remark}
\newtheorem{remark}[theorem]{Remark}

%\newcommand{\qed}{\rule{7pt}{7pt}}

%\newenvironment{proof}{\noindent{\bf Proof}\hspace*{1em}}{\qed\bigskip}
%\newenvironment{proof-sketch}{\noindent{\bf Sketch of Proof}\hspace*
%{1em}}{\qed\bigskip}
%\newenvironment{proof-idea}{\noindent{\bf Proof Idea}\hspace*{1em}}
%{\qed\bigskip}
%\newenvironment{proof-of-lemma}[1]{\noindent{\bf Proof of Lemma #1}
%\hspace*{1em}}{\qed\bigskip}
%\newenvironment{proof-attempt}{\noindent{\bf Proof Attempt}\hspace*
%{1em}}{\qed\bigskip}
%\newenvironment{proofof}[1]{\noindent{\bf Proof
%of #1:}}{\qed\bigskip}
%\newenvironment{remark}{\noindent{\bf Remark}\hspace*{1em}}{\bigskip}

%\newcommand{\figcaption}[1]{\mycaption[]{#1}}
%\newcommand{\tabcaption}[1]{\mycaption[]{#1}}
%\newcommand{\head}[1]{\chapter[Lecture \##1]{}}
%\newcommand{\mathify}[1]{\ifmmode{#1}\else\mbox{$#1$}\fi}
%%\renewcommand{\Pr}[1]{\mathify{\mbox{Pr}\left[#1\right]}}
%%\newcommand{\Exp}[1]{\mathify{\mbox{Exp}\left[#1\right]}}
%\newcommand{\bigO}O
%\newcommand{\set}[1]{\mathify{\left\{ #1 \right\}}}
%\def\half{\frac{1}{2}}

% Coding theory addenda

\newcommand{\enc}{{\sf Enc}}
\newcommand{\dec}{{\sf Dec}}
\renewcommand{\E}{{\rm Exp}}
\newcommand{\Var}{{\rm Var}}
\newcommand{\Z}{{\mathbb Z}}
\newcommand{\F}{{\mathbb F}}
\renewcommand{\K}{{\mathbb K}}
\renewcommand{\H}{{\mathbb H}}
\newcommand{\N}{{\mathbb N}}
\newcommand{\integers}{{\mathbb Z}^{\geq 0}}
\renewcommand{\R}{{\mathbb R}}
\newcommand{\Q}{{\cal Q}}
\newcommand{\calA}{{\cal A}}
\newcommand{\eqdef}{{\stackrel{\rm def}{=}}}
\newcommand{\from}{{\leftarrow}}
\newcommand{\vol}{{\rm Vol}}
\newcommand{\ip}[1]{{\langle #1 \rangle}}
\newcommand{\wt}{{\rm {wt}}}
\renewcommand{\vec}[1]{{\mathbf #1}}
\newcommand{\mspan}{{\rm span}}
\newcommand{\rs}{{\rm RS}}
\newcommand{\RM}{{\rm RM}}
\newcommand{\Had}{{\rm Had}}
\newcommand{\calc}{{\cal C}}
\newcommand{\coeff}{\mathfrak{C}}
\renewcommand{\Re}{\operatorname{Re}}
\newcommand{\ignore}[1]{}

\newcommand{\bits}{\{+1,-1\}}
\newcommand{\ftb}{: \Z_2^n \to \bits}
\newcommand{\fptb}{: \Z_p^n \to \bits}

\renewcommand{\E}{\mathop{\mathbbm{E}}}

\newcommand{\Lone}[1]{\| #1 \|_1}

\newcommand{\spar}{\mathop{\mathrm{spar}}}
\newcommand{\bias}{\mathop{\mathrm{bias}}}

\newcommand{\pdt}{{\ensuremath{\oplus}-DT}}
\newcommand{\pdtt}{{\ensuremath{\oplus}-DT. }}
\newcommand{\pD}{\ensuremath{D^\oplus}}
\newcommand{\psize}{\ensuremath{\mathrm{size}_\oplus}}

\newcommand{\ppdt}{{\ensuremath{\oplus_p}-DT}}
\newcommand{\ppdtt}{{\ensuremath{\oplus_p}-DT. }}
\newcommand{\ppD}{\ensuremath{D^{\oplus_p}}}
\newcommand{\ppsize}{\ensuremath{\mathrm{size}_{\oplus_p}}}

\newcommand{\m}[3]{m_{#1}([#2]_{#3})}
\newcommand{\s}[3]{s_{#1}([#2]_{#3})}

\newcommand*\samethanks[1][\value{footnote}]{\footnotemark[#1]}

\title{On the Structure of Boolean Functions with Small Spectral Norm}
\author{%
Amir Shpilka\thanks{Faculty of Computer Science, Technion --- Israel Institute of Technology, Haifa, Israel,
\texttt{\{shpilka,benlee\}@cs.technion.ac.il}.  The research leading to these results has received funding
from the European Community's Seventh Framework Programme (FP7/2007-2013) under grant agreement number 257575.}
\and Avishay Tal\thanks{Department of Computer Science and Applied
    Mathematics, Weizmann Institute of Science, Rehovot, Israel. Email: \texttt{avishay.tal@weizmann.ac.il}.
Research supported by an Adams Fellowship of the Israel Academy of Sciences and Humanities,
by an Israel Science Foundation grant and
by the I-CORE Program of the Planning and Budgeting Committee and the Israel Science Foundation
}
\and
Ben lee Volk\samethanks[1]
}

\date{}
\maketitle

\begin{abstract}
In this paper we prove results regarding Boolean functions with small spectral norm (the spectral norm of $f$ is $\|\hat{f}\|_1=\sum_{\alpha}|\hat{f}(\alpha)|$). Specifically, we prove the following results for functions $f:\B^n\to \B$ with $\|\hat{f}\|_1=A$.

\begin{enumerate}

\item There is a subspace $V$ of co-dimension at most $A^2$ such that $f|_V$ is constant.

\item \label{abs:main} $f$ can be computed by a parity decision tree of size $2^{A^2} n^{2A}$. (a parity decision tree is a decision tree whose nodes are labeled with arbitrary linear functions.)

\item If in addition $f$ has at most $s$ nonzero Fourier coefficients, then $f$ can be computed by a parity decision tree of depth $A^2\log s$.

\item \label{abs:learning} \sloppy For every $0<\epsilon$ there is a parity decision tree of depth $O(A^2 + \log(1/\epsilon))$ and size $2^{O(A^2)} \cdot \min\{ 1/\epsilon^2,O(\log(1/\epsilon))^{2A}\}$ that $\epsilon$-approximates $f$. Furthermore, this tree can be learned, with  probability $1-\delta$, using $\poly(n,\exp(A^2),1/\epsilon,\log(1/\delta))$ membership queries.
%\footnote{Using the result in item 1 we obtain a tree of depth $O(A^2+\log(1/\epsilon)$. The improved bound is obtained by combining our proof with the result of \cite{TsangWXZ13}.}
\end{enumerate}

All the results above also hold (with a slight change in parameters) for functions $f:\Z_p^n\to \B$.

\ignore{Green and Sanders \cite{GreenSanders08} showed that $f:\B^n\to \B$ with $\|\hat{f}\|_1=A$ can be expressed as a sum of $2^{2^{A^4}}$ characteristic vectors of subspaces. Result number \ref{abs:main} in particular implies that $f$ is a sum of at most $n^{A^2}$ characteristic functions which is a smaller number for $A>\log\log n$. Moreover, our results provide the additional structure of a parity-decision-tree. No such structure (with reasonable parameters) follow from the work of Green and Sanders.
}
\end{abstract}

%\newpage

%\thispagestyle{empty}

%\tableofcontents

%\thispagestyle{empty}

%\newpage

%\setcounter{page}{1}

%\layout

\thispagestyle{empty}
\newpage
\pagenumbering{arabic}

\section{Introduction}

The Fourier transform is one of the most useful tools in the analysis of Boolean functions. It is a household name in many areas of theoretical computer science:  Learning theory (cf. \cite{KushilevitzMansour93,LMN93,Mansour1994}); Hardness of approximation (cf. \cite{Hastad01}); Property testing (cf. \cite{BLR93,BCHKS,GOSSW:08}); Social choice (cf. \cite{KKL88,Kalai02}) and more. The  reader interested in the Fourier transform and its applications is referred to the online book \cite{OdonnellBook}.

A common theme in the study of Fourier transform is the question of classifying all Boolean functions whose Fourier transforms share some natural property. For example, Friedgut proved that Boolean functions that have {\em small influence} are close to being juntas (i.e. functions that depend on a small number of coordinates) \cite{Friedgut98}.  Friedgut, Kalai and Naor proved that Boolean functions whose Fourier spectrum is concentrated on the first two levels are close to dictator functions (i.e. functions of the form $f(x_1,\ldots,x_n)=x_i$ or $1-x_i$). In \cite{ZhangShi10,MontanaroOsborne10} it was conjectured that a Boolean function that has a {\em sparse} Fourier spectrum (i.e. that has only $s$ nonzero Fourier coefficients), can be computed by a parity decision tree (for short we denote parity decision tree by $\oplus$-DT) of depth $\poly(\log s)$. Recall that in a \pdt{} nodes are labeled by linear functions (over $\Z_2$) rather than by variables. It is well known that a function that is computed by a depth $d$ \pdt{} has sparsity at most $\exp(d)$ (see Lemma~\ref{lem:simple facts}), so this conjecture implies a (more or less) tight result. This conjecture was raised in the context of the log-rank conjecture in communication complexity and, if true, it would imply that the log-rank conjecture is true for functions of the form $F(x,y)=f(x\oplus y)$, for some Boolean function $f$.

In this paper we are interested in the structure of functions that have small spectral norm. Namely, in Boolean functions $f:\B^n\to \B$ that for some number $A$ satisfy
\begin{equation}\label{eq:maineq}
\|\hat{f}\|_1\eqdef \sum_{\alpha}|\hat{f}(\alpha)|\leq A \,,
\end{equation} where $A$ may depend on the number of variables $n$ (for definitions see Section~\ref{sec:prelim}). Such functions were studied in the context of  circuit complexity (cf. \cite{Grolmusz97}) and, more notably, in learning theory, where it is one of the most general family of Boolean functions that can be learned efficiently  \cite{KushilevitzMansour93,Mansour1994,ABFKP08}. In particular, Kushilevitz and Mansour proved that any Boolean function satisfying \eqref{eq:maineq}, can be well approximated by a sparse polynomial \cite{KushilevitzMansour93}. This already gives some rough structure for functions with small spectral norm, however one may ask for a more refined structure that captures the function exactly.  Green and Sanders were the first to obtain such a result (and until this work this was the only such result). They  proved that if $f$ satisfies Equation~\eqref{eq:maineq} then it can be expressed as a sum of at most $2^{2^{O(A^4)}}$ characteristic functions of subspaces, that is,
\begin{equation}\label{eq:GS}
f=\sum_{i=1}^{2^{2^{O(A^4)}}}\pm\ind{V_i},
\end{equation}
where each $V_i$ is a subspace.
Thus, when $A$ is constant this gives a very strong result on the structure of such a function $f$. This result can be seen as an {\em inverse} theorem, as it is well known and easy to see that the spectral norm of the characteristic function of a subspace is constant. Thus, \cite{GreenSanders08} show that in general, any function with a small spectral norm is a linear combination of a (relatively) small number of such characteristic functions. Of course, ideally one would like to show that the number of functions in the sum is at most $\poly(A)$ and not doubly exponential in $A$, however, Green and Sanders note that ``it seems to us that it would be difficult to use our method to reduce the number of exponentials below two.''

%We end this discussion by remarking that no counter example is known to the hypothesis that $f$ satisfying Equation~\eqref{eq:maineq} can be expressed as the sum of polynomially many characteristic functions of subspaces.

It is possible that another classification of Boolean functions with small spectral norm could be achieved using decision trees, or more generally, parity decision trees. It is not hard to show that if a Boolean function $g$ is computed by a \pdt{}  with $s$ leaves then the spectral norm of $g$ is at most $s$ (see Lemma~\ref{lem:simple facts}). Interestingly, we are not aware of any Boolean function that has a small spectral norm and that cannot be computed by a small \pdtt It is thus an interesting question whether this is indeed the general case, namely, that any function of small spectral norm can be computed by a small \pdtt We note that the result of \cite{GreenSanders08} does not yield such a structure. Indeed, if we were to represent the function given by Equation~\eqref{eq:GS} as a \pdt{} then, without knowing anything more about the function, then we do not see a more efficient representation than the brute-force one that yields a \pdt{} of size $n^{2^{2^{O(A^4)}}}$.

\ignore{
Recall that a decision-tree is a binary tree whose nodes are labeled by variables and its leaves by $0/1$. We treat each node as asking a query about the variable labeling it. If the answer is $0$ (i.e. the variable was assigned the value $0$) then we follow a path to the left child, and if the answer is $1$ then we go to the right child.
}

Another interesting question concerning functions with small spectral norm comes from the learning theory perspective. As mentioned above, Kushilevitz and Mansour proved that for any Boolean function satisfying Equation~\eqref{eq:maineq} there is some sparse polynomial $g=\sum_{i=1}^{A^2/\epsilon}\hat{f}(\alpha_i)\chi_{\alpha_i}(x)$ (where the coefficients in the summation are the $A^2/\epsilon$ largest Fourier coefficient of $f$) such that $\Pr_x[f(x)\neq \sgn(g(x)]\leq \epsilon$. Thus, their learning algorithm outputs as hypothesis the function $\sgn(g(x))$. This is the case even if $f$ is computed by a small decision tree or a small \pdtt It would be desirable to output a hypothesis coming from the same complexity class as $f$, i.e. to output a decision tree or a \pdtt However, a hardness result of \cite{ABFKP08} shows that under reasonable complexity assumptions, one cannot hope to output a small decision tree approximating $f$. So, a refinement of the question should be to try and output the smallest tree one can find for a function approximating $f$. For example, the function
\begin{equation}\label{eq:KM}
  \sgn(g)=\sgn\left(\sum_{i=1}^{A^2/\epsilon}\hat{f}(\alpha_i)\chi_{\alpha_i}(x)\right)
\end{equation}
can be computed by a \pdt{} of depth $O(A^2/\epsilon)$ in the natural way. Even when $A$ is a constant and $\epsilon$ is polynomially small this does not give much information. Thus, a natural question is to try and find a better representation for such a range of parameters.

\subsection{Our results}

Our first result identifies a {\em local} structure shared by Boolean functions with small spectral norm.

\begin{theorem}\label{thm:local}
Let $f:\B^n\to \B$ be such that $\|\hat{f}\|_1=A$, then, there is an affine subspace $V\subset \B^n$ of co-dimension at most $A^2$ such that $f$ is constant on $V$.
\end{theorem}

We note that the proof of \cite{GreenSanders08} does not imply the existence of such an affine subspace $V$ of such a high dimension.
%Our next two results (for functions defined over the Boolean cube) rely on the same tools as Theorem~\ref{thm:local}.
Our next result gives a \pdt{} computing $f$.

\begin{theorem}\label{thm:PDT}
Let $f:\B^n\to \B$ be such that $\|\hat{f}\|_1=A$, then, $f$ can be computed by a \pdt{} of size $2^{A^2}n^{2A}$.
\end{theorem}

In particular, the theorem implies that $f=\sum_{i=1}^{2^{A^2}n^{2A}}\pm\ind{V_i}$, where each $V_i$ is a subspace.\\

Another result settles the conjecture of \cite{ZhangShi10,MontanaroOsborne10} for the case of sparse Boolean functions with small spectral norm.

\begin{theorem}\label{thm:sparse}
Let $f:\B^n\to \B$ be such that $\|\hat{f}\|_1=A$ and $|\{\alpha\mid \hat{f}(\alpha)\neq 0\}|=s$. Then $f$ can be computed by a \pdt{} of depth $A^2\log s$.
\end{theorem}

Thus, if the spectral norm of $f$ is constant (or $\poly(\log s)$), Theorem~\ref{thm:sparse} settles the conjecture affirmatively. The conjecture is still open for the case where the spectral norm of $f$ is large.\\

Our last result (for functions over the Boolean cube) fits into the context of learning theory and provides a bound on the depth of a \pdt{} {\em approximating} a function with a small spectral norm. Here, the distance between two Boolean functions is measured with respect to the uniform distribution, namely, $\dist(f,g)=\Pr_{x\in\B^n}[f(x)\neq g(x)]$.

%An earlier version of this paper gave a bound of $O(A^3 \cdot \log(1/\epsilon))$ on the depth of the tree.

%However, Avishay Tal \cite{Tal13} noticed a simpler argument that yields a better result. Thus, the following improvement is basically due to Tal's observation.

\begin{theorem}\label{thm:approximation}
Let $f:\B^n\to \B$ be such that $\|\hat{f}\|_1=A$. Then for every $\delta, \epsilon>0$ there is a randomized algorithm that, given a query oracle to $f$, outputs (with probability at least $1-\delta$) a \pdt{} of depth $O(A^2+\log(1/\epsilon))$ and size $2^{O(A^2)}\min\{1/\epsilon^2,O(\log(1/\epsilon))^{2A}\}$, which computes a Boolean function $g_\epsilon$ such that $\dist(f,g_\epsilon)\leq \epsilon$. The algorithm runs in time polynomial in $n, \exp(A^2), 1/\epsilon$ and $\log(1/\delta)$.
\end{theorem}

Thus, when $A$ is a constant and $\epsilon$ is polynomially small, the depth is  $O(\log n)$ and the size is only poly-logarithmic in $n$. This greatly improves upon the representation guaranteed by Equation~\eqref{eq:KM}. If one insists on outputting a \pdt{}, then, for all ranges of parameters, the tree that we obtain is much smaller than the tree guaranteed by Equation~\eqref{eq:KM}.\\
%We later observe that by combining our proof techniqe with the result of \cite{TsangWXZ13}, we can replace each occurrence of $A^2$ with $A$. See Theorem~\ref{thm:better approximation} in Section~\ref{sec:TWXZ}.\\

We also prove analogs of the theorems above for functions $f \fptb$ having small spectral norm. Namely, in the theorems above %Theorems~\ref{thm:local},~\ref{thm:PDT},~\ref{thm:sparse} and \ref{thm:approximation}
one could instead talk of $f:\Z_p^n\to \B$ and obtain essentially the same results.\footnote{Of course, one would have to speak about the analog of a \pdt{} for the case where the inputs come from $\Z_p^n$.} Theorems~\ref{thm:p-local},~\ref{thm:p-PDT},~\ref{thm:p-sparse} and \ref{thm:p-approximation} are the $\Z_p$ analogs to Theorems~\ref{thm:local},~\ref{thm:PDT},~\ref{thm:sparse} and \ref{thm:approximation}, respectively. We note that in \cite{GreenSanders08b} Green and Sanders extended their result to hold for functions mapping an abelian group $G$ to $\B$, obtaining the same bound as in \cite{GreenSanders08}, so our result for functions on $\Z_p^n$ could be seen as an analog to their result for such groups.

\subsection{Comparison with \cite{GreenSanders08}}

Comparing Theorem~\ref{thm:PDT} to Equation~\eqref{eq:GS} (that was proved in \cite{GreenSanders08}), we note that while Equation~\eqref{eq:GS} does not involve the number of variables (i.e. the upper bound on the number of subspaces only involves $A$), our result does involve $n$. On the other hand, we give a more refined structure - that of a parity decision tree - which is not implied by Equation~\eqref{eq:GS} (see also the discussion above). Moreover, when $A=\Omega((\log \log n)^{1/4})$, our bound is much better than the one given in Equation~\eqref{eq:GS}.\\

Our proof technique is also quite different than that of \cite{GreenSanders08}. Their proof idea is to represent $f$ as $f=f_1+f_2$ where the Fourier supports of $f_1$ and $f_2$ are disjoint, and such that $f_1$ and $f_2$ are {\em close to being integer valued} and have a somewhat smaller spectral norm. Then, using recursion, they represent each $f_i$ as a sum of a small number of characteristic functions of subspaces. In particular, Green and Sanders do not restrict their treatment to Boolean functions but rather study functions that at every point of the Boolean cube obtain a value that is almost an integer. Thus, they prove a more general result, namely, that $f_{\Z}$, the integer part of $f$, can be represented in the form of Equation~\eqref{eq:GS}. We on the other hand only work with Boolean functions, so their result is stronger from that respect. However, while their proof was a bit involved and required using results from additive combinatorics, our approach is more elementary and is based on exploiting the fact that $f$ is Boolean. In particular, our starting point is an analysis of the simple equation $f^2=1$ (when we think of $f$ as mapping $\B^n$ to $\{\pm1\}$). Furthermore, we are able to use the fact that $f$ is Boolean in order to show that it can be computed by a small \pdt{}, which does not seem to follow from \cite{GreenSanders08}.\\

Green and Sanders later extended their technique and proved a similar result for functions over general abelian groups $f:G\to \{0,1\}$ \cite{GreenSanders08b}. Our technique do not extend to general groups, but we do obtain results for the case that $G=\Z_p^n$, which again has the same advantages and disadvantages compared to the result of \cite{GreenSanders08b} (although, the simplicity of our approach is even more evident here).

\subsection{Proof idea}

As mentioned above, our proof relies on the simple equation $f^2=1$ (when we think of $f:
\B^n\to \{\pm 1\}$). By expanding the Fourier representations (See Section~\ref{sec:prelim} for definitions) of both sides we reach the identity
$$
\sum_{\gamma} \hat{f}(\gamma)\hat{f}(\delta+\gamma) = 0,
$$
that holds for all $\delta \neq 0$ (See Lemma~\ref{lem:convolution}).
This identity could be interpreted as saying that the mass on pairs whose product is positive is the same as the mass on pairs whose product is negative. In particular, if we consider the two heaviest elements in the Fourier spectrum, say, $\hat{f}(\alpha)$ and $\hat{f}(\beta)$, and let $\delta=\alpha+\beta$, then  by restricting $f$ to one of the subspaces $\chi_\delta(x)=1$ or $\chi_\delta(x)=-1$, we get a substantial saving in the spectral norm (see Lemma~\ref{lem:restriction-reduces-norm}). This happens since there is a significant $L_1$ mass on pairs $\hat{f}(\gamma),\hat{f}(\delta+\gamma)$ that have different signs.  By repeating this process we manage to prove the existence of small \pdt{} for $f$.

The argument for functions over $\Z_p^n$ is similar, but requires more technical work. For that reason we decided to give a separate proof for the case of functions over the Boolean cube, and then, after the ideas were laid out in their simpler form, to prove the results in the more general case.

\subsection{The work of Tsang et al. \cite{TsangWXZ13}}\label{sec:TWXZ}

Independently and simultaneously to our work, Tsang et al. \cite{TsangWXZ13} obtained  related results.
% (though their work was more oriented towards proving results on the communication complexity of $F(x,y)=f(x\oplus y)$, where $f$ has small spectral norm).
The main objective of the work \cite{TsangWXZ13} was to study the communication complexity of sparse Boolean functions. These are functions $f$ such that the {\em communication matrix} of the function $F(x,y)=f(x\oplus y)$ has low rank. Resolving the log-rank conjecture from communication complexity for such functions was the main motivation for the conjecture raised in \cite{MontanaroOsborne10} and \cite{ZhangShi10}.

Tsang et al. managed to prove a stronger version of our Theorem~\ref{thm:local}, namely, they proved that $f$ is constant on a subspace of co-dimension at most $O(A)$. Their argument is identical to ours (namely, to the one given in Lemma~\ref{lem:restriction-reduces-norm}) except that they observe that after $O(1/A)$ steps of increasing the largest Fourier coefficient of $f$, it grows to at least $1/2$. From that point on they make use of the simple observation  that the proof of (their equivalent of) Lemma~\ref{lem:restriction-reduces-norm} actually guarantees that the restriction that saves the most in the spectral norm keeps increasing the largest coefficient. Thus, now at each step the spectral norm goes down by some constant factor and hence additional $O(1/A)$ many steps would make $f$ constant.\footnote{Our Lemma~\ref{lem:restriction-reduces-norm} only speaks about the spectral norm, but the effect on the largest Fourier coefficient is obvious from the proof.}

This immediately improves the results in Theorems~\ref{thm:local} and \ref{thm:sparse}; we can now change the factor $A^2$ to $A$ in both.

%\ignore{
%Combining the proof of Theorem~\ref{thm:approximation} with the result of \cite{TsangWXZ13} we conclude the following theorem. (The proof follows immediately by using the result of \cite{TsangWXZ13} in the proof of Theorem~\ref{thm:approximation} instead of our Theorem~\ref{thm:local}.)
%
%\begin{theorem}\label{thm:better approximation}
%Let $f:\B^n\to \B$ be such that $\|\hat{f}\|_1=A$. Then for every $\delta, \epsilon>0$ there is a randomized algorithm that, given a query oracle to $f$, outputs (with probability at least $1-\delta$) a \pdt{} of depth $O(A + \log(1/\epsilon))$ (hence, its size is $O(\frac{1}{\epsilon}\exp(A))$ which computes a Boolean function $g_\epsilon$ such that $\dist(f,g_\epsilon)\leq \epsilon$. The algorithm runs in polynomial time in $n, \exp(A), 1/\epsilon$ and $\log(1/\delta)$.
%\end{theorem}
%}

The work \cite{TsangWXZ13} does not contain analogs for Theorems \ref{thm:PDT} and \ref{thm:approximation}.
We also note that Tsang et al. did not study the case of functions from $\Z_p^n$ to $\B$, and so they do not have analogs of Theorems~\ref{thm:p-local}, \ref{thm:p-PDT}, \ref{thm:p-sparse} and \ref{thm:p-approximation}.

\subsection{Organization}

Section~\ref{sec:prelim} contains the basic background and definitions. In Section~\ref{sec:2} we prove our results for functions $f \ftb$. The results for functions on $\Z_p^n$ are given in Section~\ref{sec:p}. Finally, in Section~\ref{sec:open} we discuss   problems left open by this work.

%%%%%%%%%%%%%%%%%%%%%%%%%%%%%%%%%%%

\section{Notation and Basic Results}\label{sec:prelim}

It will be more convenient for us to talk about functions $f:\B^n\to \{\pm 1\}$. Note that if $f:\B^n\to \B$ then $1-2f:\B^n\to \{\pm 1\}$ and $1-2f$ and $f$ have roughly the same spectral norm (up to a multiplicative factor of $2$) and the same Fourier sparsity (up to $\pm 1$).

\subsection{Decision trees and parity decision trees}

In this section we define the basic computational models that we shall consider in the paper.

\begin{definition}[Decision tree]\label{def:dt}
A decision tree is a labeled binary tree $T$. Each internal node of $T$ is labeled with a variable $x_i$, and each leaf by a bit $b\in\bits$. Given an input $x\in\Z_2^n$, a computation over the tree is executed as follows: Starting at the root, stop if it's a leaf, and output its label. Otherwise, query its label $x_i$. If $x_i=0$, then recursively evaluate the left subtree, and if $x_i=1$, evaluate the right subtree.

\end{definition}

A decision tree $T$ computes a function $f$ if for every $x\in\Z_2^n$, the computation of $x$ over $T$ outputs $f(x)$.
The {\em depth} of a decision tree is the maximal length of a path from the root to a leaf. The decision tree complexity of $f$, denoted $D(f)$, is the depth of a minimal-depth tree computing $f$. Since one can always simply query all the variables of the input, it holds that for any Boolean function $f$, $D(f) \le n$. A comprehensive survey of decision tree complexity can be found in \cite{BuhrmanW02}.

In the context of Fourier analysis, even a function with simple Fourier spectrum, such as the parity function over $n$ bits, which has only 1 nonzero Fourier coefficient, requires a full binary decision tree for its computation, and in particular its depth is $n$.
This example suggests that a more suitable computational model for understanding the connection between the computational complexity and the Fourier expansion of a function is the {\em parity decision tree} model, first presented by Kushilevitz and Mansour (\cite{KushilevitzMansour93}).

\begin{definition}[\pdt]\label{def:pdt}
A parity decision tree is a labeled binary tree $T$, in which every internal node is labeled by a linear function $\alpha \in \Z_2^n$, and each leaf with a bit $b\in\bits$. Whenever a computation over an input $x$ arrives at an internal node, it queries $ \ip{\alpha,x} $ (where the inner product is carried modulo 2). If $\ip{\alpha,x}=0$ it recursively evaluates the left subtree, and if $\ip{\alpha,x}=1$, it evaluates the right subtree. When the computation reaches a leaf it outputs its label.

\end{definition}

Namely, a \pdt{} can make an arbitrary linear query in every internal node (and in particular, compute the parity of $n$ bits using a single query). Since a query of a single variable is linear, this model is an extension of the regular decision tree model.

The depth of the minimal-depth parity decision tree which computes $f$ is denoted $\pD(f)$, thus $\pD(f) \le D(f)$. As the example of the parity function shows, the parity decision tree model is strictly stronger than the model of decision trees. We also denote by $\psize(f)$ the size (i.e. number of leaves) of a minimal-size \pdt{}  computing $f$.

As a helpful tool, we extend the parity decision tree model to a {\em functional parity decision tree} model, in which we allow every leaf to be labeled with a Boolean function, rather than only by a constant. A functional \pdt{} $T$ then computes a function $f$ if for every leaf $\ell$ of $T$, its label equals the restriction of $f$ to the affine subspace defined by the constraints that appear on the path from $T$'s root to $\ell$.

\subsection{Fourier Transform}

We represent Boolean functions as functions $f \ftb \subseteq \R$ where $-1$ represents the Boolean value ``True'' and $1$ represents the Boolean value ``False".
For a vector of $n$ bits $\alpha$, $\alpha_i$ denotes its $i$-th coordinate.
The set of $2^n$ group characters $\left\lbrace \chi_\alpha : \Z_2^n \to \bits \mid \alpha\in \Z_2^n \right\rbrace$, with $\chi_\alpha\left(x\right)=(-1)^{\sum_{i=1}^{n} \alpha_i x_i}$ for every $\alpha \in \Z_2^n$, forms a basis of the vector space of functions from $\Z_2^n$ into $\R$. Furthermore, the basis is orthonormal with respect to the inner product\footnote{Later when we study of functions over $\Z_p^n$ we  define the inner product to be $\E_x \left[ f(x)\overline{g(x)} \right]$.}
$$
	\left\langle f,g \right\rangle = \E_x \left[ f(x)g(x) \right]
$$
where the expectation is taken over the uniform distribution over $\Z_2^n$.
The {\em Fourier expansion} of a function $f \ftb$ is its unique representation as a linear combination of those group characters:
$$
f(x) = \sum_{\alpha \in \Z_2^n} \hat{f}(\alpha) \chi_\alpha (x).
$$
Two of the basic identities of Fourier analysis, which follow from the orthonormality of the basis, are:

\begin{enumerate}
\item $\hat{f} (\alpha) = \ip{f,\chi_\alpha}=\E_x \left[ f(x)\chi_\alpha (x) \right] $
\item{(Plancherel's Theorem)} $\ip{f,g} = \E_x \left[ f(x)g(x) \right] = \sum_{\alpha \in \Z_2^n} \hat{f} (\alpha) \hat{g} (\alpha)$. % = \ip{\hat{f},\hat{g}} $
\end{enumerate}
The case $f=g$ in Plancherel's theorem is called {\em Parseval's Identity}. Furthermore, when $f$ is Boolean, $f^2 = 1$, which implies
\begin{equation}\label{parseval}
\sum_{\alpha \in \Z_2^n} \hat{f} (\alpha)^2 = 1.
\end{equation}
We define two basic complexity measures for Boolean functions:

\begin{definition}
Let $f \ftb$ be a Boolean function. The {\em sparsity} of $f$, denoted $\spar(f)$, is the number of non-zero Fourier coefficients, namely
$$ \spar(f) = \# \left\lbrace \alpha\in\Z_2^n \mid \hat{f}(\alpha) \neq 0 \right\rbrace. $$
\end{definition}
A function $f$ is said to be {\em $s$-sparse} if $\spar (f) \le s$.

\begin{definition}
Let $f \ftb$ be a Boolean function. The $L_1$ norm (also dubbed the {\em spectral norm}) of $f$ is defined as
$$ \| \hat{f} \|_1 = \sum_{\alpha \in \Z_2^n} | \hat{f} (\alpha) |. $$
\end{definition}

For every $f \ftb$ it holds that $\Lone{\hat{f}} \ge  \|f\|_\infty = 1$ (where $\|f\|_\infty=\max_{x \in \Z_2^n} |f(x)|$). We later show (Lemma \ref{lem:spectral-norm-1}) that equality is obtained if and only if $f=\pm \chi_\alpha$ for some $\alpha \in \Z_2^n$.

These measure are related to parity decision trees using the following simple lemma. For completeness we give the proof of the lemma in Appendix~\ref{app:missing lemma}.

\begin{lemma}
\label{lem:simple facts}
Let $f \ftb$ be a Boolean function computed by a \pdt{} $T$ of depth $k$ and size $m$. Then:
\begin{enumerate}
\item $\spar(f) \le m2^k \le 4^k$.
\item $\Lone{\hat{f}} \le m \le 2^k$.
\end{enumerate}
\end{lemma}

In the upcoming sections we consider restrictions of Boolean functions to (affine) subspaces of $\Z_2^n$. We denote by $f|_V$ the restriction of $f$ to a subspace $V \subseteq \Z_2^n$. For any $\alpha \neq 0$, the set $\{ x \mid \chi_\alpha (x) = 1\}$ is a subspace of $\Z_2^n$ of co-dimension 1. The restriction of $f$ to this subspace is denoted  $f|_{\chi_\alpha = 1}$. Similarly, the set $\{ x \mid \chi_\alpha (x) = -1\}$ is an affine subspace of  co-dimension 1, and we denote with $f|_{\chi_\alpha = -1}$ the restriction of $f$ to this subspace.
It can be shown (cf. \cite{OdonnellBook}, Chapter 3, Section 3.3) that under such a restriction, the coefficients $\hat{f} (\beta)$ and $\hat{f} (\alpha + \beta)$ (for every $\beta \in \Z_2^n$) collapse to a single Fourier coefficient whose absolute value is $|\hat{f} (\beta) + \hat{f} (\alpha + \beta)|$. Similarly, in the Fourier transform of $f|_{\chi_\alpha = -1}$, they collapse to a single coefficient whose absolute value is $|\hat{f}(\beta) - \hat{f} (\alpha + \beta)| $. This in particular implies that $\Lone{\hat{f}}$ and $\spar(f)$ do not increase when $f$ is restricted to such a subspace. Indeed, both facts follow easily from the representation
\begin{equation}\label{eq:cosets}
	f(x) = \sum_{\beta \in \Z_2^n/\langle \alpha \rangle} \left(\hat{f}(\beta)+\hat{f}(\beta+\alpha)\chi_\alpha(x)\right)\chi_\beta(x)\;,
\end{equation}
where $\Z_2^n/\langle \alpha \rangle$ denotes the cosets of the group $\langle \alpha \rangle=\{0,\alpha\}$ in $\Z_2^n$.
When studying a restricted function, say $f'=f|_{\chi_\alpha(x)=1}$, we shall abuse notation and denote with
$\widehat{f'}(\beta)$ the term corresponding to the coset $\beta + \langle \alpha \rangle$. Namely,
$\widehat{f'}(\beta) = \hat{f}(\beta)+\hat{f}(\beta+\alpha)$. (similarly, for $f''=f|_{\chi_\alpha(x)=-1}$, we shall denote $\widehat{f''}(\beta) = \hat{f}(\beta)-\hat{f}(\beta+\alpha)$.)
Thus, in $f'$ both $\widehat{f'}(\beta)$ and $\widehat{f'}(\beta+\alpha)$ refer to the same Fourier coefficient as we only consider coefficients modulo $\langle \alpha \rangle$ (similarly for $f''$).

%%%%%%%%%%%%%%%%%%%%%%%%%%%%%%%%%%%

\section{Boolean functions with small spectral Norm}\label{sec:2}

In this section we prove our main results for functions over the Boolean cube.
While many of the proofs and techniques used for general primes also apply to the case $p=2$, we find the case $p=2$ substantially simpler, so we present the proofs for this case separately.

%%%%%%%%%%%%%%%%%%%%%%%%%%%%%%%%%%%
\subsection{Basic tools}
%%%%%%%%%%%%%%%%%%%%%%%%%%%%%%%%%%%

In this section we prove the following lemma, which states that for every Boolean function $f \ftb$, with small spectral norm, there exists a linear function $\chi_\gamma$ such that both restrictions $f|_{\chi_{\gamma}=1}$ and $f|_{\chi_{\gamma}=-1}$ have noticeable smaller spectral norms compared to $f$. In Section~\ref{sec:p} we give a generalization of the lemma for functions $f:\Z_p^n\to \{+1,-1\}$ (Lemma~\ref{lem:restriction-in-Fp}).

\begin{lemma}[Main Lemma for functions over $\Z_2^n$]
\label{lem:restriction-reduces-norm}
Let $f \ftb$ be a Boolean function. Let $\hat{f} (\alpha)$ be $f$'s maximal Fourier coefficient in absolute value, and $\hat{f} (\beta)$ be the second largest, and suppose $\hat{f}(\beta) \neq 0$.
Let $f'=f|_{\chi_{\alpha+\beta}=1}$ and $f''=f|_{\chi_{\alpha+\beta}=-1}$. Then, if $\hat{f}(\alpha)\hat{f}(\beta) > 0$ then it holds that
$$
\Lone{\hat{f'}} \le \Lone{\hat{f}} - |\hat{f}(\alpha)|
\quad \text{and}\quad
\Lone{\hat{f''}} \le \Lone{\hat{f}} -|\hat{f} (\beta) |.$$
If $\hat{f}(\alpha)\hat{f}(\beta) < 0$ then
$$
\Lone{\hat{f'}} \le \Lone{\hat{f}} - |\hat{f}(\beta)|
\quad \text{and}\quad
\Lone{\hat{f''}} \le \Lone{\hat{f}} -|\hat{f} (\alpha) |.$$
\end{lemma}

The proof of the lemma follows from analyzing  the simple equation $f^2=1$.

\begin{lemma}
\label{lem:convolution}
Let $f \ftb$ be a Boolean function. For all $\alpha \neq 0$, it holds that
$$
\sum_{\gamma} \hat{f}(\gamma)\hat{f}(\alpha+\gamma) = 0.
$$
\end{lemma}

\begin{proof}
Since $f$ is Boolean we have that $f^2 = 1$. In the Fourier representation,
$$
\left( \sum_{\gamma} \hat{f}(\gamma)\chi_\gamma (x) \right)
\left( \sum_{\beta} \hat{f}(\beta)\chi_\beta (x) \right) = 1.
$$
Then $\sum_{\gamma} \hat{f}(\gamma)\hat{f}(\alpha+\gamma)$ is the Fourier coefficient $\widehat{f^2}(\alpha)$ of the function $f^2$ at $\alpha$. However, if $\alpha \neq 0$ then this coefficient equals 0 by the uniqueness of the Fourier expansion of the function $f^2=1$.
\end{proof}

\begin{proof}[Proof of Lemma \ref{lem:restriction-reduces-norm}]
Without loss of generality assume that $\hat{f}(\alpha)\hat{f}(\beta) > 0$, i.e. they have the same sign (the other case is completely analogous.) By Lemma \ref{lem:convolution},
\begin{equation}
\label{eq:convolution}
\sum_{\gamma\in\Z_2^n}\hat{f}(\gamma)\hat{f}(\alpha+\beta+\gamma)=0.
\end{equation}
Let $N_{\alpha+\beta} \subseteq \Z_2^n$ be the set of vectors $\gamma$ such that $\hat{f}(\gamma)\hat{f}(\alpha+\beta+\gamma) < 0$ (Note that by assumption, $\alpha, \beta \not \in N_{\alpha+\beta}$). Switching sides in \eqref{eq:convolution}, we get:
$$
2\left| \hat{f} (\alpha) \hat{f}(\beta) \right|=
\sum_{\gamma \in N_{\alpha+\beta}}
\left| \hat{f}(\gamma)\hat{f}(\alpha+\beta+\gamma) \right| -
\sum_{\substack{\gamma \not\in N_{\alpha+\beta} \\ \gamma\neq\alpha,\beta}} \left|\hat{f}(\gamma)\hat{f}(\alpha+\beta+\gamma) \right|.
$$
In particular,
\begin{equation}
\label{ineq:two-largest-coeffs}
| \hat{f} (\alpha) | | \hat{f}(\beta) | \le \frac{1}{2} \sum_{\gamma \in N_{\alpha+\beta}}
\left| \hat{f}(\gamma)\hat{f}(\alpha+\beta+\gamma) \right|.
\end{equation}
We now use the fact that that $\hat{f}(\beta)$ is the second largest in absolute value, and $\hat{f}(\alpha)$ does not appear in the sum, to bound the right hand side:
\begin{equation}
\label{ineq:sum-of-coeffs}
\sum_{\gamma \in N_{\alpha+\beta}}
\left| \hat{f}(\gamma)\hat{f}(\alpha+\beta+\gamma) \right|
\le
|\hat{f}(\beta)| \sum_{\gamma \in N_{\alpha+\beta}} \min
\left\lbrace
|\hat{f}(\gamma)|, |\hat{f}(\alpha+\beta+\gamma)| \right \rbrace.
\end{equation}
Then \eqref{ineq:two-largest-coeffs} and \eqref{ineq:sum-of-coeffs} (as well as the assumption $|\hat{f}(\beta)|>0$) together imply
\begin{equation}
\label{ineq:largest-coeff}
| \hat{f} (\alpha) | \le \frac{1}{2}
\sum_{\gamma \in N_{\alpha+\beta}} \min
\left\lbrace
|\hat{f}(\gamma)|, |\hat{f}(\alpha+\beta+\gamma)| \right \rbrace.
\end{equation}
Let $f'=f|_{\chi_{\alpha+\beta}=1}$. Then for every $\gamma$ the coefficients $\hat{f} (\gamma) $ and $\hat{f} (\alpha+\beta+\gamma)$ collapse to a single coefficient whose absolute value is
$ | \hat{f} (\gamma) + \hat{f} (\alpha+\beta+\gamma) | $ (recall Equation~\eqref{eq:cosets}). For $\gamma \in N_{\alpha+\beta}$, $$
|\hat{f}(\gamma) + \hat{f}(\alpha+\beta+\gamma)| =
\left\vert |\hat{f}(\gamma) | - |\hat{f}(\alpha+\beta+\gamma)| \right\vert
$$
which reduces the $L_1$ norm of $f'$ compared to that of $f$ by at least $\min(|\hat{f}(\gamma) | , |\hat{f}(\alpha+\beta+\gamma)|)$. In total, since both $\gamma$ and $\alpha+\beta+\gamma$ belong to $ N_{\alpha+\beta}$,
%is counted twice in the right hand side of \eqref{ineq:largest-coeff},
we get:
$$
\Lone{\widehat{f'}} \le \Lone{\hat{f}} - \frac{1}{2} \sum_{\gamma \in N_{\alpha+\beta}}
\min
\left\lbrace
|\hat{f}(\gamma)|, |\hat{f}(\alpha+\beta+\gamma)| \right \rbrace.
$$
Therefore by \eqref{ineq:largest-coeff} we have
$$
\Lone{\widehat{f'}} \le \Lone{\hat{f}} - |\hat{f}(\alpha)|.
$$
When we consider $f''=f|_{\chi_{\alpha+\beta}=-1}$ we clearly have that for $\gamma=\alpha$,
$$
|\widehat{f''}(\gamma)| =  | \hat{f} (\gamma) -\hat{f} (\alpha+\beta+\gamma) | = | \hat{f} (\alpha)| -|\hat{f} (\beta) |.$$
Hence,
$$
\Lone{\widehat{f''}} \le \Lone{\hat{f}} -|\hat{f} (\beta) |.$$
%By Lemma \ref{lem:largest-coeff}, $\Lone{\hat{f'}} \le \Lone{\hat{f}} - 1/A$.
\end{proof}

Next, we show that any Boolean function with small spectral norm has a large Fourier coefficient.

\begin{lemma}
\label{lem:largest-coeff}
Let $f \ftb$ be a Boolean function. Denote $A = \Lone{\hat{f}}$, and let $\hat{f} (\alpha)$ be $f$'s maximal Fourier coefficient in absolute value. Then $|\hat{f}(\alpha)| ~ \ge ~ 1/A$. Furthermore, let $\hat{f} (\beta)$ be $f$'s second largest Fourier coefficient in absolute value. Then $|\hat{f}(\beta)| > (1-\hat{f}(\alpha)^2)/\Lone{\hat{f}} = (1-\hat{f}(\alpha)^2)/A$.
\end{lemma}

\begin{proof}
By Parseval's identity,
$$ 1 = \E [ f^2 ] = \sum_{\gamma} \hat{f} (\gamma)^2.  $$
Now note that $$  1 = \sum_{\gamma} \hat{f} (\gamma)^2 \le |\hat{f}(\alpha)| \sum_{\gamma} |\hat{f}(\gamma)| \le A |\hat{f}(\alpha)| ,$$
which implies that indeed  $|\hat{f}(\alpha)| \ge 1/A$.
The second statement follows similarly, since
$$
1-\hat{f}(\alpha)^2 = \sum_{\gamma\neq\alpha} \hat{f} (\gamma)^2 \le |\hat{f}(\beta)| \sum_{\gamma \neq \alpha} |\hat{f}(\gamma)| < \Lone{\hat{f}}\cdot |\hat{f}(\beta)| = A |\hat{f}(\beta)|.
$$
\end{proof}

\begin{corollary}\label{cor:restriction-reduces-norm}
Let $f \ftb$ be a Boolean function such that $ \| \hat{f} \|_1 = A>1$. Then there exists $\gamma \in \Z_2^n$ and $b \in \bits$ such that $ \Lone{\widehat{f|_{\chi_\gamma = b}}} \le A - 1/A$.
\end{corollary}

\begin{proof}
The assumption $A>1$ implies the second largest coefficient, $\hat{f}(\beta)$, is non-zero, and then the result is immediate from Lemma~\ref{lem:restriction-reduces-norm} and Lemma~\ref{lem:largest-coeff}.
\end{proof}

%%%%%%%%%%%%%%%%%%%%%%%%%%%%%%%%%%%
\subsection{Proofs of Theorems}
%%%%%%%%%%%%%%%%%%%%%%%%%%%%%%%%%%%

We now show how Theorems~\ref{thm:local},\ref{thm:PDT},\ref{thm:sparse} and \ref{thm:approximation} follow as simple consequences of Lemma~\ref{lem:restriction-reduces-norm}.

\begin{lemma}
\label{lem:spectral-norm-1}
Let $f \ftb$ be a Boolean function such that $\Lone{\hat{f}} = 1$. Then $f = \pm \chi_\alpha$ for some $\alpha \in \Z_2^n$.
\end{lemma}

\begin{proof}
By Parseval's identity and the assumption, we get
$$
 \sum_{\gamma} \hat{f} (\gamma)^2 = 1 = \sum_{\gamma} |\hat{f} (\gamma)|.
$$
For all $\gamma$ we have that $|\hat{f} (\gamma)| \in [0,1]$, so $|\hat{f} (\gamma)| < \hat{f} (\gamma)^2$ unless $|\hat{f} (\gamma)| = 1$ or $\hat{f}(\gamma) = 0$, and the proposition follows.
\end{proof}

Corollary \ref{cor:restriction-reduces-norm} and Lemma \ref{lem:spectral-norm-1} imply Theorem \ref{thm:local}:

\begin{proof}[Proof of Theorem \ref{thm:local}]
Apply Corollary \ref{cor:restriction-reduces-norm} iteratively on $f$. After less than $A^2$  steps, we are left with a function $g$ which is a restriction of $f$ on an affine subspace defined by the restrictions so far, such that $\Lone{\hat{g}} = 1$. By Lemma \ref{lem:spectral-norm-1}, $g = \pm \chi_\alpha$ for some $\alpha \in \Z_2^n$. If $\alpha \neq 0$ we further restrict $g$ on $\chi_\alpha = 1$ to get a restriction of $f$ which is constant.
\end{proof}

We note that the proof of Theorem~\ref{thm:local} actually implies that $f$ is constant on a subspace of co-dimension at most ${A+1 \choose 2}$. As mentioned earlier, a slight twist in the proof improves the co-dimension to $O(A)$ \cite{TsangWXZ13}.

\begin{proof}[Proof of Theorem \ref{thm:PDT}]

Let
$$
L(n,A) \eqdef \max_{\substack{f\ftb \\ \Lone{\hat{f}} \le A}} \psize(f).
$$
We show, by induction on $n$, that $L(n,A) \le 2^{A^2} \cdot n^{2A}$.

For $n=1$ the result is trivial.

Let $n>1$ and further assume that $A>1$ (if $A=1$ then the claim follows from Lemma~\ref{lem:spectral-norm-1}).
Let $\hat{f}(\alpha),\hat{f}(\beta)$ be the first and second largest Fourier coefficients in absolute value, respectively.
%By Corollary \ref{cor:restriction-reduces-norm}, there is a linear function $\gamma \in \Z_2^n$ and $b \in \bits$ such that
%$ \Lone{\widehat{f|_{\chi_\gamma = b}}} \le A - 1/A$.
By Lemma~\ref{lem:largest-coeff} we are in one of the following cases:
\begin{enumerate}
\item $|\hat{f}(\alpha)| \ge 1/2$
\item $1/2 > |\hat{f}(\alpha)| \ge 1/A$ and $|\hat{f}(\beta)| > \frac{1-\hat{f}(\alpha)^2}{A} \ge \frac{3}{4A}$.
\end{enumerate}

Consider the  tree whose first query is the linear function $\chi_\gamma$ where $\gamma = \alpha + \beta$
(i.e. we branch left or right according to the value of $\langle x,\gamma\rangle$).
By the choice of $\gamma$ we obtain the following recursion:
In case 1, $$L(n,A) \le L(n-1,A-1/2) + L(n-1, A);$$ While in case 2, $$L(n,A) \le L(n-1,A-1/A) + L(n-1, A-3/(4A)).$$
Note also that in the second case $A \ge 2$, or else $|\hat{f}(\alpha)| \ge 1/2$ by Lemma~\ref{lem:largest-coeff}.
Induction follows in the first case as
\begin{align*}
L(n,A) &\leq L(n-1,A-1/2) + L(n-1,A)\\
        &\leq 2^{(A-1/2)^2} \cdot (n-1)^{2(A-1/2)} + 2^{A^2} \cdot (n-1)^{2A}  \\
		&\leq 2^{A^2} \cdot (n-1)^{2(A-1/2)} \left( 1 + (n - 1) \right) \\
		&< 2^{A^2} \cdot n^{2A}\;.
\end{align*}
In the second case we have
\begin{align*}
L(n,A)  &\leq L(n-1,A-1/A) + L(n-1,A-3/(4A))\\
        &\leq 2^{(A-1/A)^2} \cdot (n-1)^{2(A-1/A)} + 2^{(A-3/(4A))^2} \cdot (n-1)^{2(A-3/(4A))}  \\
		&\leq 2^{A^2} \cdot n^{2A} \left( 2^{-2+1/A^2} + 2^{-3/2+(3/4)^2/A^2} \right) \\
        &\leq 2^{A^2} \cdot n^{2A}\;,
\end{align*}
where in the last inequality we used the fact that $A \ge 2$.
\end{proof}

As the AND function demonstrates, this argument gives a result that is tight up to a polynomial factor in some cases.

\begin{proof}[Proof of Theorem \ref{thm:sparse}]
By Theorem \ref{thm:local}, there exist $A^2$ linear functions $\alpha_1,\ldots,\alpha_{A^2}$ that can be fixed to values $b_1,\ldots,b_{A^2}$, respectively, where $b_i \in \bits$ for $1\le i \le A^2$, such that $f$ restricted to the subspace $\{x\mid \chi_{\alpha_i}(x)=b_i \;,\; \forall 1\le i \le A^2\}$ is constant. This implies that for any non-zero coefficient $\hat{f}(\beta)$ there exists at least one other non-zero coefficient $\hat{f} (\beta + \gamma)$ for $\gamma \in \text{span}\{\alpha_1,\ldots,\alpha_{A^2}\}$. Indeed, if no such coefficient exists then the restriction $f|_{\chi_{\alpha_1}(x)=b_1,\ldots,\chi_{\alpha_{A^2}}=b^2}$ will have the non-constant term $\hat{f}(\beta)\cdot \chi_\beta$ (for example, this can be easily obtained from Equation~\eqref{eq:cosets}).
Therefore, for any other fixing of $\chi_{\alpha_1},\ldots,\chi_{\alpha_{A^2}}$, both $\hat{f}(\beta)\chi_\beta$ and  $\hat{f} (\beta + \gamma)\chi_{\beta+\gamma}$ collapse to the same (perhaps non-zero) linear function, which implies that $\spar(f|_{\chi_{\alpha_1}=b'_1,\ldots,\chi_{\alpha_{A^2}}=b'_{A^2}}) \le \spar(f)/2$ for any choice of $b'_1,\ldots,b'_{A^2}$. In other words, if we consider the tree of depth $A^2$ in which on level $i$ all nodes branch according to $\ip{\alpha_i,x}$ then restricting $f$ to any path yields a new function with half the sparsity. Thus, we can continue this process by induction for at most $\log s$ steps, until all the functions in the leaves are constant. The resulting tree has depth at most $A^2\log s$ as claimed.
\end{proof}

Our next goal is proving Theorem \ref{thm:approximation}. To this end, we use a lemma which shows there exists a low depth functional \pdt{} which computes a function $g$ such that $\Pr_x [f(x) \neq g(x)] \le \epsilon$, where $x$ is drawn from the uniform distribution over $\Z_2^n$. Recall that the {\em bias} of a Boolean function $f$ is defined to be
$$
\bias(f) \eqdef \left| \Pr_x[f(x)=1] - \Pr_x[f(x)=-1] \right|.
$$
Alternatively, $\bias(f) = |\hat{f}(0)|$.

\begin{lemma}
\label{lem:functional-pdt}
\sloppy Let $f \ftb$ be a Boolean function with $\Lone{\hat{f}} \leq A$. Then, there exists a functional \pdt{} of depth at most $O(A^2 + \log (1/\epsilon))$ that computes a function $g$ such that $\Pr_x [f(x) \neq g(x)] \le \epsilon$. Furthermore, the size of the tree is at most $2^{O(A^2)}\min\{1/\epsilon^2,O(\log(1/\epsilon))^{2A}\}$.
\end{lemma}

\begin{proof}
Let $K=\max\left\lbrace 10A^2, 2\log(1/\epsilon)\right\rbrace$ be a bound on the depth of the tree. In order to construct the functional decision tree, we use a recursive argument that stops whenever we reach a constant leaf, or after $K$ levels of recursion, and then show that for a uniformly random $x \in \Z_2^n$, $x$ arrives at a highly biased leaf with probability $\geq 1-\epsilon$, hence proving the statement of the lemma.

Let $\hat{f}(\alpha)$ be $f$'s largest coefficient in absolute value, and $\hat{f}(\beta)$ the second largest. Note that if $|\hat{f}(0)| > 1-\epsilon$ we are done. Hence, we consider two cases:

\begin{enumerate}
\item $|\hat{f}(\alpha)| > 1-\epsilon$ for $\alpha \neq 0$:

We first show that if $|\hat{f}(\alpha)| > 1-\epsilon$ then $|\hat{f}(0)| < \epsilon$. By considering $-f$ instead of $f$, if needed, we may assume without the loss of generality $\hat{f}(\alpha) > 1-\epsilon$. Note that
$$
1-\epsilon < \hat{f}(\alpha) = \Pr[f=\chi_\alpha] - \Pr [f\neq\chi_\alpha] = (1-\Pr[f\neq\chi_\alpha]) - \Pr[f\neq\chi_\alpha],
$$
so $\Pr[f\neq\chi_\alpha] < \epsilon/2$. Now, since $\E[\chi_\alpha] = 0$, we have
$$
|\hat{f}(0)| = |\E[f]| =|\E[f] - \E[\chi_\alpha]| = |\E[f-\chi_\alpha]| \le \E[|f-\chi_\alpha|] = 2\Pr[f\neq\chi_\alpha] < \epsilon.
$$
In this case we query on $\chi_\alpha$. Note that no matter what value $\chi_\alpha$ obtains, the restricted function has bias at least $|\hat{f}(\alpha)|-|\hat{f}(0)|>1-2\epsilon$, and we terminate the recursion.

\item $|\hat{f}(\alpha)| \le 1-\epsilon$:

In this case we query on $\chi_{\alpha+\beta}$. Let $f'=f|_{\chi_{\alpha+\beta}=1}$ and $f''=f|_{\chi_{\alpha+\beta}=-1}$ By Lemma \ref{lem:restriction-reduces-norm}, for at least one of $f'$ and $f''$, the spectral norm drops by at least $1/A$. We continue by induction the construction on $f'$ and $f''$, terminating when all the leaves are highly biased (in particular this includes the case of a constant leaf), or after at most $K$ levels of recursion.
\end{enumerate}

It remains to be shown that the fraction of inputs $x \in \Z_2^n$ that arrive at an unbiased leaf is at most $\epsilon$. We say an internal node labeled $\chi_\gamma$ is norm-reducing for $x$, if $\chi_{\gamma}(x)=b$ and the restriction on $\chi_\gamma=b$ reduces the spectral norm by at least $1/A$. Clearly, a computation over any input $x$ which traverses $A^2$ norm reducing nodes for $x$ arrives at a constant leaf. Furthermore, by construction, all the leaves which are not highly biased appear in the $K$-th level of the tree. Hence, an input which arrives at an unbiased node satisfies $K$ independent linear equations, for which at most $A^2$ are norm reducing. Since for every fixed $0 \neq \gamma \in \Z_2^n$ and $b \in \bits$ the probability that $\chi_\gamma(x)=b$ is exactly $1/2$, the probability that $x$ arrives at a non highly biased node is bounded by\footnote{We count how many words in $\{0,1\}^K$ with fewer than $A^2$ $1$'s are there.}
$$
\sum_{i=0}^{A^2-1} \frac{{K \choose i}}{2^K} \le
2^{-K} A^2 {K \choose A^2} \le 2^{(-K/2)} \le \epsilon
$$
by the choice of $K$.

To prove the upper bound on the size of the tree we first note that $2^K$ is a trivial upper bound. Moreover, as in the proof of Theorem~\ref{thm:PDT}, the construction that we have satisfies the recursion formula
\begin{align*}
S(K-d,B) \le & \max \left\lbrace S(K-(d+1),B-1/2) + S(K-(d+1), B), \right. \\
        & \left. S(K-(d+1),B-1/B) + S(K-(d+1), B-3/(4B)) \right\rbrace,
\end{align*}
where $S(K-d,B)$ stands for the number of leaves in the tree rooted at a node $v$ at depth $d$ such that the function $f_v$ computed at $v$ satisfies $\Lone{\hat{f_v}}\leq B$. As before, the solution to this recursion is $S(K,A)\leq 2^{A^2}K^{2A}$.
Overall, we have that the size of the tree the approximating parity decision is at most:
\begin{align*}
\min\left\lbrace{ 2^K, 2^{A^2}K^{2A}}\right\rbrace
&= \min\left\lbrace \max\left\lbrace{2^{10A^2}, \epsilon^{-2}}\right\rbrace, 2^{A^2} \cdot \max\left\lbrace{ (2\log(1/\epsilon))^{2A}, (10A^2)^{2A}}\right\rbrace  \right\rbrace\\
&\le \min\left\lbrace 2^{10A^2} \cdot \epsilon^{-2}, 2^{A^2} \cdot (10A^2)^{2A} \cdot (2\log(1/\epsilon))^{2A}\right\rbrace\\
&\le 2^{O(A^2)} \cdot \min \left\lbrace \epsilon^{-2}, O(\log(1/\epsilon))^{2A} \right\rbrace
\end{align*}
as claimed.

%$${K \choose A^2}\leq \left(\frac{eK}{A^2}\right)^{A^2}\leq \left(2+\frac{\log(1/\epsilon)}{A^2}\right)^{O(A^2)} .$$
\end{proof}

%By combining this argument with Lemma 30 of \cite{TsangWXZ13} et al. we can improve the upper bound on the depth to $O(A+\log(1/\epsilon)$.

Note that if we replace each highly biased function-labeled leaf in the functional \pdt{} from Lemma \ref{lem:functional-pdt} with the constant it is biased towards (i.e. by the sign of its constant term), the total error would increase by at most $\epsilon$. That is, it can be easily converted to a regular \pdt{} of a function $g$ which $\epsilon$-approximates $f$.
In fact, in the proof of Lemma~\ref{lem:functional-pdt}, we could have continued the recursion until reaching a constant leaf or depth $K$, but for the sake of understanding the proof of Theorem~\ref{thm:approximation} it may be more clear to keep the current version in mind.

The proof of Theorem \ref{thm:approximation} follows by combining Lemma~\ref{lem:functional-pdt} with the well known result of Goldreich and Levin \cite{GoldreichLevin89} and of Kushilevitz and Mansour \cite{KushilevitzMansour93}, who showed that given a query oracle to a function $f$, with high probability, one can approximate its large Fourier coefficients in polynomial time.

\begin{lemma}[\cite{GoldreichLevin89,KushilevitzMansour93}]
\label{alg:KM}
There exists a randomized algorithm, such that given a query oracle to a function $f \ftb$, and parameters $\delta, \theta,\eta$, outputs, with probability at least $1-\delta$, a list containing all of $f$'s Fourier coefficients whose absolute value is at least $\theta$. Furthermore, the algorithm outputs an additive approximation of at most $\eta$ to each of these coefficients. The algorithm runs in polynomial time in $n$, $1/\theta$, $1/\eta$ and $\log(1/\delta)$.
\end{lemma}

\begin{proof}[Proof of Theorem \ref{thm:approximation}]
We use the algorithm from Lemma \ref{alg:KM} to find $f$'s largest Fourier coefficient in absolute value, $\hat{f}(\alpha)$. Whenever $|\hat{f}(\alpha)|\le 1-\epsilon$, Lemma \ref{lem:largest-coeff} implies $|\hat{f}(\beta)| > \frac{1-\hat{f}(\alpha)^2}{\Lone{\hat{f}}} \ge \frac{1-\hat{f}(\alpha)^2}{A}> \epsilon/A$, so the same algorithm can be used to find the second largest coefficient, $\hat{f}(\beta)$, in time $\poly(n,A,1/\epsilon,\log (1/\delta))$. We use Lemma \ref{lem:functional-pdt} to construct a functional \pdt, and replace every function-labeled leaf with the constant it's biased towards. The bound on the running time follows from the size of the \pdt{} and the running time of the algorithm from Lemma \ref{alg:KM}.

In fact, there is a slight inaccuracy in the argument above. Note that Lemma~\ref{alg:KM} only guarantees that we find a coefficient that is approximately the largest one. However, if it is the case that the second largest coefficient is very close to the largest one, then in Lemma~\ref{lem:functional-pdt} when we branch according to $\chi_{\alpha+\beta}$ both children have significantly smaller spectral norm.

If it is the case that we correctly identified the largest Fourier coefficient but failed to identify the second largest then we note that if our approximation is good enough, say better than $\epsilon/2A$, then even if we are mistaken and branch according to $\chi_{\alpha+\beta'}$ where $\left| |\hat{f}(\beta)|-|\hat{f}(\beta')|\right| < \epsilon/2A$, the the argument in Lemma~\ref{lem:functional-pdt} still works, perhaps with a slightly worse constant in the big O.
\end{proof}

%%%%%%%%%%%%%%%%%%%%%%%%%%%%%%%%%%%

\section{Functions over $\Z_p^n$ with small spectral norm}\label{sec:p}

In this section, we extend our results to functions $f \fptb$ where $p$ is any fixed prime.   Throughout this section we assume $p>2$. We start by giving some basic facts on the Fourier transform over $\Z_p^n$.

\subsection{Preliminaries}

Let $\omega=e^{\frac{2\pi i}{p}} \in \C$ be a primitive root of unity of order $p$. The set of $p^n$ group characters $$
\{\chi_\alpha:\Z_p^n \to \C \mid \alpha \in \Z_p^n\}
$$ where $\chi_\alpha (x) = \omega^{\ip{\alpha,x}}$, is a basis for the vector space of functions from $\Z_p^n$ into $\C$, and is orthonormal with respect to the inner product $\ip{f,g} = \E_x [f(x)\overline{g(x)}]$.\footnote{For a complex number $z$, we denote by $\overline{z}$ its complex conjugate.} We now have that $\hat{f}(\alpha) = \E_x[f(x)\overline{\chi_\alpha(x)}]$ and $f=\sum_{\alpha\in\Z_p^n}\hat{f}(\alpha)\chi_\alpha$. Plancherel's theorem holds here as well and the sparsity and $L_1$ norm are defined in the same way as they were defined for functions $f \ftb$.
%, and the inversion formula $\hat{f}(\alpha) = \E_x[f(x)\overline{\chi_\alpha(x)}]$ and Plancherel's Theorem generalize to this case as well.
Lemma \ref{lem:largest-coeff} also extends to functions $f \fptb$, with virtually the same proof. When $f$ is real-valued (and in particular, a Boolean function), then $\hat{f}(0)=\E[f]$ is real, and it can also be directly verified that $\hat{f}(\alpha) = \overline{\hat{f}(-\alpha)}$.

We have the analog to Equation~\eqref{eq:cosets}:
\begin{equation}\label{eq:p-cosets}
	f(x) = \sum_{\beta \in \Z_p^n/\langle \alpha \rangle} \left(\sum_{k=0}^{p-1}\hat{f}(\beta+k\cdot \alpha)(\chi_\alpha(x))^k\right)\chi_\beta(x)\;.
\end{equation}
Hence, when $f$ is restricted to an affine subspace on which $\chi_\alpha = \omega^\lambda$ (where $0\le \lambda \le p-1$), then for every\footnote{Recall that $\langle \alpha \rangle$ is the additive group generated by $\alpha$ and $\Z_p^n / \langle \alpha \rangle$ is the set of cosets of $\langle \alpha \rangle$.} $\beta \in \Z_p^n / \langle \alpha \rangle$ we have
$$
\hat{g}(\beta) = \sum_{k=0}^{p-1} \omega^{\lambda k} \hat{f}(\beta + k\alpha).
$$
For every $\beta \in \Z_p^n$, we denote by $[\beta]_\alpha = \beta + \langle \alpha \rangle$ the coset of $\langle \alpha \rangle$ in which $\beta$ resides.

Lemma \ref{lem:convolution} now becomes:
\begin{equation}
\label{eq:convolution-p}
\sum_{\alpha \in \Z_p^n} \hat{f}(\alpha) \hat{f}(\beta - \alpha) = 0
\end{equation}
for all $0 \neq \beta \in \Z_p^n$.

As a generalization of the \pdt{} model, we define a $p$-ary linear decision tree, denoted $\oplus_p$-DT, to be a computation tree where every internal node $v$ is labeled by a linear function $\gamma \in \Z_p^n$ and has $p$ children. The edges between $v$ and its children are labeled $0,1,\ldots,p-1$, and on an input $x$, it computes $\ip{\gamma,x} \mod p$ and branches accordingly. We carry along from the binary case the notation $\ppD(f)$ and $\ppsize(f)$, and define them to be the depth (respectively, size) of a minimal-depth (resp. size) \ppdt{} computing $f$.

\subsection{Basic tools}

In this section we prove the basic tools required for generalizing the theorems for functions defined on $\Z_2^n$
%Theorems \ref{thm:local}, \ref{thm:PDT} and \ref{thm:sparse} for
to functions $f \fptb$. As a generalization of Lemma \ref{lem:restriction-reduces-norm}, we show a slightly more complex and detailed argument:

%\begin{lemma}
%\label{lem:restriction-in-Fp}
%Let $f \fptb$ be a non-constant Boolean function such that $\Lone{\hat{f}} = A$. Then there exist a constant $c = c(p) > 0$, $0\neq\gamma \in \Z_p^n$ and at least $p-1$ distinct elements $\lambda_1,\ldots,\lambda_{p-1} \in \Z_p$ such that $\Lone{\widehat{f|_{\chi_\gamma=\omega^{\lambda_k}}}} \le A - c/A$ for $k=1,\ldots,p-1$.
%\end{lemma}
%
%\begin{lemma}
%\label{lem:restriction-Fp-second-largest}
%Let $f \fptb$ be a non-constant Boolean function such that $\Lone{\hat{f}} = A$. Let $\hat{f}(\alpha)$ be its largest coefficient in absolute value, and $\hat{f}(\beta)$ be the second largest. Then there exist $0 \neq \gamma \in \Z_p$, $\lambda_0 \in \Z_p$ and a universal constant $c>0$ such that $\Lone{\widehat{f|_{\chi_\gamma=\omega^{\lambda_0}}}} \le A - c|\hat{f}(\alpha)| \le A- c/A$. Furthermore, for all $\lambda \in \Z_p$, $\Lone{\widehat{f|_{\chi_\gamma=\omega^\lambda}}} \le A - c|\hat{f}(\beta)|$.
%\end{lemma}

\begin{lemma}[Main Lemma for functions over $\Z_p^n$]
\label{lem:restriction-in-Fp}
Let $f \fptb$ be a non-constant Boolean function such that $\Lone{\hat{f}} = A$.
Let $\hat{f}(\alpha)$ be its largest coefficient in absolute value, and $\hat{f}(\beta)$ be the second largest.
Then there exist %$0 \neq \gamma \in \Z_p^n$,
a universal constant $c_0$ and a constant $c_1 = c_1 (p)=O(1/p^2)$ such that
\begin{enumerate}
\item \label{item:all_reduce_f_beta} For all $\lambda \in \Z_p$, $\Lone{\widehat{f|_{\chi_{\beta-\alpha}=\omega^\lambda}}} \le A - c_0|\hat{f}(\beta)|$.
\item \label{item:one_reduce_f_alpha} There exists at least $m:=\lfloor{p/3\rfloor}$ distinct elements $\lambda_1, \ldots, \lambda_m \in \Z_p$ such that $\Lone{\widehat{f|_{\chi_{\beta-\alpha}=\omega^{\lambda_k}}}} \le A - c_0|\hat{f}(\alpha)| \le A - c_0 / A$ for $k=1,\ldots,m$.
\item \label{item:p-1_reduce_f_alpha} There exists at least $p-1$ distinct elements $\lambda_1,\ldots,\lambda_{p-1} \in \Z_p$ such that
$\Lone{\widehat{f|_{\chi_{\beta-\alpha}=\omega^{\lambda_k}}}} \le A - c_1 \cdot |\hat{f}(\alpha)| \le A-c_1/A$ for $k=1,\ldots,p-1$.
\end{enumerate}
\end{lemma}

As before we first prove a claim characterizing functions with very small spectral norm. Observe that when $p>2$, the characters themselves are not Boolean functions any more. The following is a variant of Lemma \ref{lem:spectral-norm-1} for $\Z_p^n$ with $p>2$.

\begin{lemma}
\label{lem:spectral-norm-p-1}
Let $f \fptb$ be a Boolean function such that $\Lone{\hat{f}} = 1$. Then $f = \pm 1$.
\end{lemma}

\begin{proof}
Once more, using Parseval's identity and the assumption:
$$
 \sum_{\gamma} |\hat{f} (\gamma)|^2 = 1 = \sum_{\gamma} |\hat{f} (\gamma)|.
$$
As before, $|\hat{f} (\gamma)| \in [0,1]$, which implies $|\hat{f} (\alpha)| = 1$ for exactly one $\alpha \in \Z_p^n$, i.e. $f = z\cdot \chi_\alpha$ where $z\in\C$ and $|z|=1$. Since $f$ is Boolean and $f(0) = z$, we get $z=\pm1$, and $\pm \chi_\alpha$ is Boolean (when $p>2$) only when $\alpha=0$.
\end{proof}

The following is a purely geometric lemma we use in our analysis. Since the Fourier coefficients now are complex numbers we need to bound the decrease in the spectral norm when two coefficients that are not aligned in the same direction collapse to the same coefficient.

\begin{lemma}
\label{lem:triangle-inequality}
Let $z_1,z_2 \in \C$ such that $|z_1| = R$, $|z_2|=r$ and $r \le R$. Suppose the angle between $z_1$ and $z_2$ is
 $\theta$. Then, for $C=C(\theta) = (1-\cos(\theta)) / 2$ it holds that
$$|z_1| + |z_2| - |z_1 + z_2| \ge Cr.$$
\end{lemma}

We give the simple proof in Appendix~\ref{app:triangle}.\\

The next lemma is similar to the inequalities of the type we used in the proof of Lemma \ref{lem:restriction-reduces-norm}.
\begin{lemma}
\label{lem:largest-coeff-inequality-Fp}
Let $f \fptb$ be a non-constant Boolean function, and suppose $\hat{f}(0)$ is the largest Fourier coefficient in absolute value and $\hat{f}(\beta)$ is the second largest. Then
$$
2|\hat{f}(0)| \le \sum_{\substack{\gamma \in \Z_p^n \\ \gamma \neq 0,\beta}}
\min \left\lbrace |\hat{f}(\gamma)|, |\hat{f}(\gamma-\beta)| \right\rbrace.
$$
\end{lemma}
\begin{proof}
By rearranging Equation \eqref{eq:convolution-p} with respect to $\beta$, we get:
$$
|2\hat{f}(0)\hat{f}(\beta)| = \left| \sum_{\substack{\gamma \in \Z_p^n \\ \gamma \neq 0,\beta}} \hat{f}(\gamma)\hat{f}(\beta-\gamma) \right|
$$
Now apply the triangle inequality to the right hand side, and then utilize the fact that $\hat{f}(\beta)$ is the second largest in absolute value and $\hat{f}(0)$ does not appear in the right hand side, to obtain
$$2|\hat{f}(0)||\hat{f}(\beta)| \le |\hat{f}(\beta)| \sum_{\substack{\gamma \in \Z_p^n \\ \gamma \neq 0,\beta}}
\min \left\lbrace |\hat{f}(\gamma)|, |\hat{f}(\beta-\gamma)| \right\rbrace.$$
Since $f$ is real-valued, $\hat{f}(\beta-\gamma) = \overline{\hat{f}(\gamma-\beta)}$ (and in particular, they have the same absolute value), and since $f$ is non-constant, by Lemma \ref{lem:spectral-norm-p-1} we have $\Lone{\hat{f}} > 1$, i.e. $\hat{f}(\beta)\neq 0$, which implies the desired inequality.
\end{proof}

When analyzing the loss in the $L_1$ norm which is caused by restriction on $\chi_\myeta$, it will be convenient to sum over the individual losses on pairs $\hat{f}(\mygamma), \hat{f}(\mygamma+\myeta)$ that collapse to the same coefficient. However, letting $\mygamma$ run over all of $\Z_p^n$, these pairs are not pairwise disjoint, so we might over-count the losses. The following lemma generously accounts for such over-counting issues, by showing that summing over all (not pairwise disjoint) pairs differs from the true counting by at most a constant factor.

\begin{lemma}
\label{lem:overcounting}
Let $f \fptb$, $0 \neq \myeta \in \Z_p^n$, and $\lambda \in \Z_p$.
If
$$
\sum_{\mygamma \in \Z_p^{n}}{|\hat{f}(\mygamma)| + |\hat{f}(\myeta+\mygamma)| - |\hat{f}(\mygamma) + \omega^{\lambda} \hat{f}(\myeta+\mygamma)|} = B,
$$
then
$$
\sum_{\mygamma \in \Z_p^n/\langle \myeta \rangle} \left(\sum_{k=0}^{p-1}{\left|\hat{f}(\mygamma+k \myeta)\right|} - \left|\sum_{k=0}^{p-1}\hat{f}(\mygamma+k \myeta)\omega^{\lambda k}\right|\right) \ge B/3.
$$
\end{lemma}
Note that the left hand side of the last inequality is exactly the loss in the $L_1$ norm when restricting $f$ on $\chi_{\myeta} = \omega^\lambda$. We defer the proof of Lemma \ref{lem:overcounting} to Appendix \ref{app:overcounting}.

\begin{lemma}
\label{lem:p-two-largest-close}
Let $f \fptb$ be a non-constant Boolean function such that $\Lone{\hat{f}} = A$. Let $\hat{f}(\alpha)$ be its largest coefficient in absolute value, and $\hat{f}(\beta)$ be the second largest. Suppose $\lambda \in \Z_p$ is such that the angle between $\hat{f}(\alpha)$ and $\omega^{\lambda} \hat{f}(\beta)$ in absolute value is at most $\pi / 3$. Then there is a universal constant $c>0$ such that $\Lone{\widehat{f|_{\chi_{\beta-\alpha}=\omega^\lambda}}} \le A - c|\hat{f}(\alpha)| \le A- c/A$.
\end{lemma}

\begin{proof}
Denote $\myeta \eqdef \beta-\alpha$. Under the assumption of the lemma, noting that\footnote{$\Re(z)$ is the real part of a complex number $z$ and $\overline{z}$ is its  conjugate.} $\hat{f}(\alpha)\cdot \overline {\omega^{\lambda}\hat{f}(\beta)} = \hat{f}(\alpha)\cdot \omega^{-\lambda}\hat{f}(-\beta)$
\begin{equation}
 \label{eq:realpart-big}
\Re\left(\hat{f}(\alpha)\cdot \omega^{-\lambda}\hat{f}(-\beta)\right) \ge \cos(\pi/3) \cdot |\hat{f}(\alpha)\cdot \omega^{-\lambda}\hat{f}(-\beta)| = \cos(\pi/3) |\hat{f}(\alpha)||\hat{f}(\beta)|.
\end{equation}
By equation \eqref{eq:convolution-p}, with respect to $-\myeta \neq 0$, we have
$$
\sum_{\mygamma}{\hat{f}(\mygamma)\hat{f}(-\myeta-\mygamma)} = 0.
$$
Hence
$$
c_0\cdot\hat{f}(\alpha)\hat{f}(-\beta) = -\sum_{\mygamma \neq \alpha,-\beta }{\hat{f}(\mygamma)\hat{f}(-\myeta-\mygamma)},
$$
where $c_0=1$ if $\alpha=-\beta$ and $c_0=2$ otherwise. Multiplying by $\omega^{-\lambda}$ and taking the real part of both sides gives
\begin{equation}
\label{eq:realpart-convolution}
\Re\left(c_0\cdot\omega^{-\lambda}\hat{f}(\alpha)\hat{f}(-\beta)\right) = \sum_{\mygamma \neq \alpha,-\beta}{-\Re\left(\hat{f}(\mygamma) \omega^{-\lambda} \hat{f}(-\myeta-\mygamma)\right)}.
\end{equation}
Let $N_{\myeta} = \left\{ \mygamma \mid \Re \left(\hat{f}(\mygamma) \omega^{-\lambda} \hat{f}(-\myeta-\mygamma)\right)<0, \mygamma \neq \alpha,-\beta \right\}$. Then \eqref{eq:realpart-convolution}, as well as the fact that $c_0 \in \{1,2\}$ and the left hand side is positive (by \eqref{eq:realpart-big}), imply
\begin{equation}
\label{eq:realpart-convolution-upperbound}
\Re\left(\omega^{-\lambda}\hat{f}(\alpha)\hat{f}(-\beta)\right) \le \sum_{\mygamma \in N_{\myeta}}{-\Re\left(\hat{f}(\mygamma) \omega^{-\lambda} \hat{f}(-\myeta-\mygamma)\right)}.
\end{equation}
Note that for every pair $(\hat{f}(\mygamma)$, $\omega^{-\lambda} \hat{f}(-\myeta-\mygamma))$, where $\mygamma \in N_{\myeta}$, the angle between $\hat{f}(\mygamma)$ and $\overline{\omega^{-\lambda}\hat{f}(-\myeta-\mygamma)}$ ($=\omega^{\lambda} \hat{f}(\myeta + \mygamma)$) is in the range $[\pi/2,3\pi/2]$. Furthermore, when applying the restriction $\chi_{\myeta} = \omega^{\lambda}$ each such pair collapses to the same coefficient, and since the angle between the two coefficients (in absolute value) is at least $\pi/2$, by Lemma \ref{lem:triangle-inequality} it follows that for all $\mygamma \in N_{\myeta}$,
\begin{equation}
\label{eq:p-save-in-one-coefficient}
|\hat{f}(\mygamma)| + |\hat{f}(\myeta+\mygamma)| - |\hat{f}(\mygamma) + \omega^{\lambda} \hat{f}(\myeta+\mygamma)| \ge \min\{|\hat{f}(\mygamma)|,|\hat{f}(\myeta+\mygamma)|\} \cdot C(\pi/2).
\end{equation}
(where $C(\theta) = (1-\cos(\theta)) / 2$ is as defined in Lemma \ref{lem:triangle-inequality}). Using \eqref{eq:p-save-in-one-coefficient} to bound the loss on every coefficient, and summing over all $\mygamma \in \Z_p^n$ (while bearing in mind that every summand is non-negative), we have
\begin{align*}
\sum_{\mygamma}&{|\hat{f}(\mygamma)| + |\hat{f}(\myeta+\mygamma)| - |\hat{f}(\mygamma) + \omega^{\lambda} \hat{f}(\myeta+\mygamma)|} \\
    &\ge \sum_{\mygamma \in N_{\myeta}}{|\hat{f}(\mygamma)| + |\hat{f}(\myeta+\mygamma)| - |\hat{f}(\mygamma) + \omega^{\lambda} \hat{f}(\myeta+\mygamma)|} \\
    &\ge \sum_{\mygamma \in N_{\myeta}}{\min\{|\hat{f}(\mygamma)|,|\hat{f}(\myeta+\mygamma)|\} \cdot C(\pi/2)}
\end{align*}
Since neither $\alpha$ nor $-\beta$ appear in $N_\myeta$, and $\hat{f}(\beta)$ is the second largest coefficient, for all $\mygamma \in N_\myeta$ it holds that
$$
\min\{|\hat{f}(\mygamma)|,|\hat{f}(\myeta+\mygamma)|\} \ge \frac{|\hat{f}(\mygamma)| |\hat{f}(\myeta+\mygamma)|}{|\hat{f}(\beta)|},
$$
so
\begin{equation}
\label{eq:p-sum-all-vectors}
\sum_{\mygamma}{|\hat{f}(\mygamma)| + |\hat{f}(\myeta+\mygamma)| - |\hat{f}(\mygamma) + \omega^{\lambda} \hat{f}(\myeta+\mygamma)|}
\ge \frac{1}{|\hat{f}(\beta)|} \cdot C(\pi/2) \cdot \sum_{\mygamma \in N_{\myeta}} |\hat{f}(\mygamma)| |\hat{f}(\myeta+\mygamma)|.
\end{equation}
Taking the complex conjugate and multiplying by $\omega^{-\lambda}$, it is also clear that for all $\mygamma$
$$
|\hat{f}(\mygamma)| |\hat{f}(\myeta+\mygamma)| = |\hat{f}(\mygamma)| |\omega^{-\lambda} \hat{f}(-\myeta-\mygamma)| = |\hat{f}(\mygamma) \omega^{-\lambda} \hat{f}(-\myeta-\mygamma)|
\ge -\Re\left(\hat{f}(\mygamma) \omega^{-\lambda} \hat{f}(-\myeta-\mygamma)\right).
$$
Hence \eqref{eq:p-sum-all-vectors} and \eqref{eq:realpart-convolution-upperbound} imply
\begin{align*}
\sum_{\mygamma}{|\hat{f}(\mygamma)| + |\hat{f}(\myeta+\mygamma)| - |\hat{f}(\mygamma) + \omega^{\lambda} \hat{f}(\myeta+\mygamma)|} & \ge
\frac{C(\pi/2) }{|\hat{f}(\beta)|} \cdot
\sum_{\mygamma \in N_{\myeta}}{-\Re\left(\hat{f}(\mygamma) \omega^{-\lambda} \hat{f}(-\myeta-\mygamma)\right)} \\
&\ge \frac{C(\pi/2) }{|\hat{f}(\beta)|} \Re\left(\omega^{-\lambda}\hat{f}(\alpha)\hat{f}(-\beta)\right).
\end{align*}
And \eqref{eq:realpart-big} now gives
$$
\sum_{\mygamma}{|\hat{f}(\mygamma)| + |\hat{f}(\myeta+\mygamma)| - |\hat{f}(\mygamma) + \omega^{\lambda} \hat{f}(\myeta+\mygamma)|} \ge
\frac{C(\pi/2) \cdot \cos(\pi/3) \cdot |\hat{f}(\alpha)|\cdot|\hat{f}(\beta)|}{|\hat{f}(\beta)|}  \ge c|\hat{f}(\alpha)|
$$
where $c$ is an absolute constant.
By Lemma~\ref{lem:overcounting} the $L_1$ norm of the restricted function has decreased by at least $c|\hat{f}(\alpha)| / 3$.
\end{proof}

%\begin{proof}[Proof of Lemma \ref{lem:restriction-Fp-second-largest}]
%Let $\myeta = \beta - \alpha$ as in Lemma \ref{lem:p-two-largest-close}. Let $\lambda\in \Z_p$, and consider the restriction $\chi_\gamma = \omega^\lambda$. Let $\theta$ be the angle between $\hat{f}(\alpha)$ and $\omega^{\lambda}\hat{f}(\alpha + \gamma) = \omega^\lambda \hat{f}(\beta)$. If $\theta$ is larger in absolute value than $\pi/3$, then under the restriction the coefficients $\hat{f}(\alpha)$ and $\hat{f}(\beta)$ collapse into the same coefficient, resulting (by Lemma \ref{lem:triangle-inequality} in a $C(\pi/3) \cdot |\hat{f}(\beta)|$ loss in the $L_1$ norm (where $C(\cdot)$ is as stated in Lemma~\ref{lem:triangle-inequality}). If, however $\theta \le \pi/3$, Lemma \ref{lem:p-two-largest-close} implies a loss of $c |\hat{f}(\alpha)|$ in the $L_1$ norm where $c$ is an absolute constant. Furthermore, since multiplication by $\omega$ rotates $\hat{f}(\beta)$ by $2\pi/p \le 2\pi/3$, there exists at least one value for $\gamma \in \Z_p$ such that $|\theta|$ would be at most $\pi/3$.
%\end{proof}

We are now ready to prove the main lemma for functions over $\Z_p^n$.
\begin{proof}[Proof of Lemma \ref{lem:restriction-in-Fp}.]
Let $\myeta = \beta - \alpha$ as in Lemma \ref{lem:p-two-largest-close}.
Let $\lambda\in \Z_p$, and consider the restriction $\chi_\myeta = \omega^\lambda$.
Let $\theta$ be the angle between $\hat{f}(\alpha)$ and $\omega^{\lambda}\hat{f}(\alpha + \myeta) = \omega^\lambda \hat{f}(\beta)$.
If $\theta$ is larger in absolute value than $\pi/3$, then under the restriction the coefficients $\hat{f}(\alpha)$ and $\hat{f}(\beta)$ collapse into the same coefficient, resulting by Lemma \ref{lem:triangle-inequality} in a $C(\pi/3) \cdot |\hat{f}(\beta)|$ loss in the $L_1$ norm (where $C(\cdot)$ is as stated in Lemma~\ref{lem:triangle-inequality}).
If $\theta \le \pi/3$ then Lemma \ref{lem:p-two-largest-close} implies a loss of $c_0 |\hat{f}(\alpha)|$ (which is also at least $c_0 |\hat{f}(\beta)|$) in the $L_1$ norm where $c_0$ is an absolute constant. This completes Item~\ref{item:all_reduce_f_beta} in the proof.

Furthermore, since multiplication by $\omega$ rotates $\hat{f}(\beta)$ by $2\pi/p$, there exists at least $\lfloor{p/3\rfloor}$ values for $\myeta \in \Z_p$ such that $|\theta|$ would be at most $\pi/3$, which completes Item~\ref{item:one_reduce_f_alpha} in the proof. \\

Next, we prove Item~\ref{item:p-1_reduce_f_alpha}.
Let $C = C(\pi/p) = (1-\cos(\pi/p)) / 2 = O(1/p^2)$ as in Lemma~\ref{lem:triangle-inequality}. We distinguish between two cases:
The first case we consider is $\alpha \neq 0$. In this case, by the fact the $f$ is real-valued $|\hat{f}(-\alpha)| = |\hat{f}(\alpha)|$. % \ge 1/A$.
So $\beta=-\alpha$ and by Item~\ref{item:all_reduce_f_beta}, restricting on $\chi_\myeta =\omega^{\lambda}$, for any $\lambda \in \Z_p$, yields
\begin{equation}
\label{eq:nonzerolargest-reduce}
\Lone{\widehat{f|_{\chi_\myeta=\omega^\lambda}}} \le \Lone{\hat{f}} - c_0\cdot |\hat{f}(\alpha)|,
\end{equation}
which implies Item~\ref{item:p-1_reduce_f_alpha} for this case.

The second case is that the largest Fourier coefficient in absolute value is achieved on $\alpha = 0$. In this case $\beta \neq 0$, and $\eta=\beta$.
By the assumption $\Lone{\hat{f}}>1$, we have $|\hat{f}(\beta)| > 0$.
We define the {\em weight} of a pair $\{\gamma, \gamma - \beta\} \subseteq \Z_p^n$ to be $w(\gamma) = \min \left\lbrace |\hat{f}(\gamma)|, |\hat{f}(\gamma-\beta)| \right\rbrace$, and denote
$$
W = \sum_{\gamma \neq 0,\beta} w(\gamma).
$$
Thus By Lemma \ref{lem:largest-coeff-inequality-Fp}, we have
\begin{equation}
\label{ineq:largest-coeff-sum-weights}
2|\hat{f}(0)|\le W.
\end{equation}
Note that when restricting $f$ on $\chi_\beta=\omega^\lambda$, $\hat{f}(\gamma)$ and $\hat{f}(\gamma - \beta)$ collapse to the same coefficient.
The new coefficient becomes $|\hat{f}(\gamma) + \omega^{\lambda}\hat{f}(\gamma + \beta)+ \cdots + \omega^{\lambda(p-1)}\hat{f}(\gamma + (p-1)\beta)|$. We analyze only the loss in the $L_1$ norm obtained from the collapse of $\hat{f}(\gamma)$ and $\hat{f}(\gamma+(p-1)\beta)=\hat{f}(\gamma-\beta)$ to the same coefficient. Let $\theta$ be the angle between $\hat{f}(\gamma)$ and $\hat{f}(\gamma-\beta)$. Since multiplication by $\omega$ is equivalent to rotation by $2\pi/p$, as $\lambda$ traverses over $0,1,...,p-1$, the angle between $\hat{f}(\gamma)$ and $\omega^{\lambda(p-1)}\hat{f}(\gamma-\beta)$ attains all possible values $\theta + 2\kappa\pi/p$  for $\kappa=0,1,...,p-1$. Hence, there exists at most one choice of $\lambda$ such that the angle between $\hat{f}(\gamma)$ and $\omega^{\lambda(p-1)}\hat{f}(\gamma-\beta)$ is less than $\pi/p$.
  We call $\lambda\in\Z_p$ {\em good} with respect to $\gamma$ if the angle between $\hat{f}(\gamma)$ and $\omega^{\lambda(p-1)} \hat{f}(\gamma - \beta)$ is at least $\pi/p$. If we fix $\beta$, then for every pair there exist at least $p-1$ good elements in $\Z_p$.
Intuitively, each element $\lambda$ which is good guarantees a loss of at least $C\cdot\min\{|\hat{f}(\gamma)|,|\hat{f}(\gamma-\beta)|\}=Cw(\gamma)$ in the spectral norm (the actual analysis, which will now follow, is a bit more delicate).

Consider now the matrix $M$ whose rows are indexed by elements $\gamma \in \Z_p^n$ for all $\gamma \neq 0,\beta$, and whose columns are indexed by all elements $\lambda \in \Z_p$. We define:
$$
M_{\gamma,\lambda} = \begin{cases}
w(\gamma) & \text{if \ensuremath{\lambda} is good with respect to \ensuremath{\gamma}} \\
0 & \text{otherwise}
\end{cases}\;.
$$
Since for every $\gamma$ there are at least $p-1$ good elements, we have
\begin{equation}
\label{ineq:total-weight}
\sum_{\gamma,\lambda} M_{\gamma,\lambda}\ge (p-1) \sum_{\gamma \neq 0,\beta} w(\gamma)= (p-1)W.
\end{equation}
While for every fixed column $\lambda_0$,
\begin{equation}
\label{ineq:weight-in-column}
\sum_{\gamma} M_{\gamma,\lambda_0} \le W.
\end{equation}
As there are $p$ columns, \eqref{ineq:total-weight} and \eqref{ineq:weight-in-column} together imply that there is at most one column in which the total weight is less than $W/2$, i.e. for all $\lambda\in\Z_p$ but at most one, it holds that
\begin{equation}
\label{ineq:good-element}
\sum_{\gamma} M_{\gamma,\lambda} \ge W/2\;.
\end{equation}
Every element $\lambda\in \Z_p$ which satisfies \eqref{ineq:good-element} will be called {\em good}. We thus proved the existence of at least $p-1$ good elements $\lambda$.

We now fix a good element $\lambda$ and consider the restriction $\chi_\beta=\omega^\lambda$.
By Lemma~\ref{lem:overcounting}, the loss of the spectral norm under this restriction is at least
$$
1/3 \cdot \sum_{\gamma}{|\hat{f}(\gamma)| + |\hat{f}(\gamma - \beta)| - |\hat{f}(\gamma) + \omega^{\lambda (p-1)}\hat{f}(\gamma-\beta)|},
$$
which is, by Lemma~\ref{lem:triangle-inequality} and the definition of $M_{\lambda,\gamma}$, at least
$$
1/3 \cdot \sum_{\gamma: \text{$\lambda$ is good w.r.t. $\gamma$}} {C\cdot\min \left\lbrace |\hat{f}(\gamma)|, |\hat{f}(\gamma-\beta)| \right\rbrace} = 1/3 \cdot \sum_{\gamma}{C \cdot M_{\gamma, \lambda}} \ge C \cdot W /6 \ge |\hat{f}(0)| \cdot C/3,
$$
where we used \eqref{ineq:good-element} and \eqref{ineq:largest-coeff-sum-weights} for the penultimate and last inequalities, respectively. Letting $c_1 = C/3$ completes the proof of the lemma.
\end{proof}

\subsection{Analogs of Theorems \ref{thm:local}, \ref{thm:PDT}, \ref{thm:sparse} and \ref{thm:approximation}}

Theorems \ref{thm:p-local}, \ref{thm:p-PDT}, \ref{thm:p-sparse} and \ref{thm:p-approximation} now follow as consequences of Lemma \ref{lem:restriction-in-Fp}. Their proofs use the same arguments we used to deduce their $\Z_2^n$ counterparts from Lemma \ref{lem:restriction-reduces-norm}. We use the notation $O_p( \cdot )$ when the underlying constant depends on $p$, whereas when we use $O(\cdot)$, the underlying constant is some absolute constant.

\begin{theorem}\label{thm:p-local}
Let $f \fptb$ be a Boolean function with $\Lone{\hat{f}} = A$. Then there exists an affine subspace $V\subseteq\Z_p^n$ of co-dimension at most $O(A^2)$ such that $f$ is constant on $V$.
\end{theorem}

\begin{proof}
Apply Lemma \ref{lem:restriction-in-Fp} iteratively on $f$. By assumption $p > 2$, so $\left\lfloor p/3 \right\rfloor \ge 1$, and then using Item \ref{item:one_reduce_f_alpha} in the proof, after at most $A^2/c_0$  steps, we are left with a function $g$ which is a restriction of $f$ on an affine subspace defined by the restrictions so far, such that $\Lone{\hat{g}} = 1$. By Lemma \ref{lem:spectral-norm-p-1} $g = \pm 1$.
\end{proof}

\begin{theorem}\label{thm:p-PDT}
Let $f \fptb$ be a Boolean function with $\Lone{\hat{f}} = A$. Then $\ppsize(f) \le p^{O_p(A^2)} n^{O_p(A)}$.
\end{theorem}

\begin{proof}
By Lemma \ref{lem:restriction-in-Fp}, there is a constant $0<c\le 1$ (where $c:=\min\{c_0,c_1\}$ depends only on $p$), a linear function $\gamma \in \Z_p^n$ and $\lambda_1,...,\lambda_{p-1} \in \Z_p$ such that $ \Lone{\widehat{f|_{\chi_\gamma = \omega^{\lambda_j}}}} \le A - c|\hat{f}(\alpha)|$ for all $1\le j \le p-1$. Furthermore, for the $p$-th direction $\lambda_p$, the same lemma shows that
$\Lone{\widehat{f|_{\chi_\gamma = \omega^{\lambda_p}}}} \le A - c|\hat{f}(\beta)|$ where $\hat{f}(\beta)$ is the second largest coefficient.

As before, let
$$
L(n,A) \eqdef \max_{\substack{f\fptb \\ \Lone{\hat{f}} \le A}} \ppsize{f}.
$$
We show, by induction on $n$, that $L(n,A) \le p^{2A^2/c} n^{2A/c}$. For $n=1$ the result is trivial.
Let $n>1$ and further assume that $A>1$ (if $A=1$ then the claim follows from Lemma~\ref{lem:spectral-norm-p-1}).

As in the proof of Theorem \ref{thm:PDT}, we consider two cases:
\begin{enumerate}
\item $|\hat{f}(\alpha)| \ge 1/2$
\item $1/2 > |\hat{f}(\alpha)| \ge 1/A$ and $|\hat{f}(\beta)| > \frac{1-|\hat{f}(\alpha)|^2}{A} \ge \frac{3}{4A}$.
\end{enumerate}

We again Consider the  tree whose first query is the linear function $\chi_\gamma$.
In the first case, by the choice of $\gamma$ we obtain the following recursion:
$$
L(n,A) \le (p-1)L(n-1,A-c/2) + L(n-1,A).
$$
The induction hypothesis then implies (using the assumption that $A>1$)
\begin{align*}
L(n,A) &\leq (p-1) p^{2(A-c/2)^2/c} (n-1)^{2(A-c/2)/c} + p^{2A^2/c} (n-1)^{2A/c}  \\
		&\leq (p-1) p^{2(A^2/c) - 1} (n-1)^{2A/c - 1} + p^{2A^2/c} (n-1)^{2A/c} \\
		&\leq p^{2A^2/c} (n-1)^{2A/c - 1} \left( 1 + (n - 1) \right) \\
		&\leq p^{2A^2/c} n^{2A/c}.
\end{align*}
While in the second case, we have the recurrence
$$
L(n,A) \le (p-1)L(n-1,A-c/A) + L(n-1,A-3c/(4A)) \le p \cdot L(n-1,A-3c/(4A))
$$
Again, the induction hypothesis implies (using the assumption that $A > 1$)
\begin{align*}
L(n,A) &\leq p \cdot p^{2(A-3c/(4A))^2/c} (n-1)^{2(A-(3c/(4A)))/c}  \\
		&\leq n^{2A/c} \cdot p^{1 + 2A^2/c - 3 + 18c/(16A^2)}   \\
        &\leq n^{2A/c} \cdot p^{2A^2/c}\;.
\end{align*}
%
%\begin{align*}
%L(n,A) &\leq (p-1) p^{(A-c/A)^2/c} (n-1)^{2(A-c/A)/c} + p^{(A-3c/(4A))^2/c} (n-1)^{2(A-(3c/(4A)))/c}  \\
%		&\leq n^{2A/c} \cdot p^{A^2/c} \left((p-1)p^{-2 c/A^2} + p^{-3/2+9c/(16A^2)}  \right) \\
%		&\leq n^{2A/c} \cdot p^{A^2/c} \left((p-1)p^{-2 + c/A^2} + p^{-3/2+9c/(16A^2)}  \right) \\
%		&\leq n^{2A/c} \cdot p^{A^2/c} \left(p^{-1 + c/A^2} + p^{-3/2+9c/(16A^2)}  \right) \\
%        &\leq n^{2A/c} \cdot p^{A^2/c}.
%\end{align*}
\end{proof}
As an immediate corollary, we get:
\begin{corollary}
Let $f \fptb$ be a Boolean function with $\Lone{\hat{f}} = A$. Then $f=\sum_{i=1}^{p^{O_p(A^2)}n^{O_p(A)}}\pm\ind{V_i}$, where each $V_i$ is an affine subspace of $\Z_p^n$.

\end{corollary}

\begin{theorem}
\label{thm:p-sparse}
Let $f \fptb$ be such that $\|\hat{f}\|_1=A$ and $|\{\alpha\mid \hat{f}(\alpha)\neq 0\}|=s$. Then $f$ can be computed by a \ppdt{} of depth $O(A^2\log s)$.
\end{theorem}

\begin{proof}
By Theorem \ref{thm:p-local}, there exist $K=O(A^2)$ linear functions $\alpha_1,\ldots,\alpha_K$ which can be fixed to values $\omega^{\lambda_1},\ldots,\omega^{\lambda_K}$ where $\lambda_j \in \Z_p$ for $1\le j \le K$, such that $f$ restricted to the subspace $\{x\mid \chi_{\alpha_j}(x)=\omega^{\lambda_j} \;,\; \forall 1\le j \le K\}$ is constant. Once again, this implies that for any non-zero coefficient $\hat{f}(\beta)$ there exists at least one other non-zero coefficient $\hat{f} (\beta + \gamma)$ for $\gamma \in \text{span}\{\alpha_1,\ldots,\alpha_K\}$, since if no such coefficient exists then the restriction $f|_{\chi_{\alpha_1}(x)=\omega^{\lambda_1},\ldots,\chi_{\alpha_{K}}=\omega^{\lambda_k}}$ will have the non-constant term $\hat{f}(\beta)\cdot \chi_\beta$.
Therefore, for any other fixing of $\chi_{\alpha_1},\ldots,\chi_{\alpha_K}$, both $\hat{f}(\beta)\chi_\beta$ and  $\hat{f} (\beta + \gamma)\chi_{\beta+\gamma}$ collapse to the same (perhaps non-zero) linear function, which implies that $\spar(f|_{\chi_{\alpha_1}=\omega^{\lambda'_1},\ldots,\chi_{\alpha_{K}}=\omega^{\lambda'_K}}) \le \spar(f)/2$ for any choice of $\lambda'_1,\ldots,\lambda'_K$. Thus, we can continue by induction until all the functions in the leaves are constant.
\end{proof}

\begin{theorem}\label{thm:p-approximation}
Let $f \fptb$ be such that $\|\hat{f}\|_1=A$. Then for every $\delta, \epsilon>0$ there is a randomized algorithm that given a query oracle to $f$ outputs (with probability at least $1-\delta$) a \ppdt{} of depth $O(A^2 + \log(1/\epsilon))$ and size $\min\{p^{O(A^2 + \log (1/\epsilon))},p^{O_p(A^2)} \cdot O(A^2 + \log(1/\epsilon))^{O_p(A)}\}$ that computes a Boolean function $g_\epsilon$ such that $\dist(f,g_\epsilon)\leq \epsilon$. The algorithm runs in polynomial time in $n, \exp(A^2), 1/\epsilon$ and $\log(1/\delta)$.
\end{theorem}

The proof of Theorem \ref{thm:p-approximation} follows the same outline as the proof of Theorem \ref{thm:approximation}. A {\em functional \ppdt} is defined as a \ppdt{} where we allow every leaf to be labeled by a Boolean function on $\Z_p^n$, and the bias of a function $f \fptb$ is defined as in the binary case. We again show there exists a low depth $\pdt$ which computes a function $g$ such that $\Pr_x [f(x) \neq g(x)] \le \epsilon$ (where $x$ is drawn from the uniform distribution over $\Z_p^n$.

\sloppy
\begin{lemma}
\label{lem:p-functional-pdt}
Let $f \ftb$ be a Boolean function with $\Lone{\hat{f}} \leq A$. Then there exists a functional \pdt{} of depth at most $O(A^2 + \log (1/\epsilon))$ and size
$\min\{p^{O(A^2 + \log (1/\epsilon))},p^{O_p(A^2)} \cdot O(A^2 + \log(1/\epsilon))^{O_p(A)}\}$, computing a function $g$ such that $\Pr_x [f(x) \neq g(x)] \le \epsilon$.
\end{lemma}
\fussy

\begin{proof}
The proof is essentially the same as the proof of Lemma \ref{lem:functional-pdt}, taking $K = \frac{20}{c_0}(A^2 + \log(1/\epsilon))$, where $c_0$ is as in Lemma \ref{lem:restriction-in-Fp}.
Note that in this case, if $|\hat{f}(\alpha)| > 1-\epsilon$, then since $|\hat{f}(\alpha)| = |\hat{f}(-\alpha)|$, by Parseval's identity,
if $\epsilon < (1-1/\sqrt{2})$ then this can only happen if $\alpha = 0$, hence $f$ is already highly biased.
Furthermore, for a random $x \in \Z_p^n$ and fixed $\gamma \in \Z_p^n$, Lemma \ref{lem:restriction-in-Fp} implies a node labeled $\chi_\gamma$ is norm reducing (by an absolute constant $c_0 / A$) for $x$ with probability $\frac{\lfloor p/3 \rfloor}{p} \ge 1/5$, hence a similar argument to the one used in Lemma \ref{lem:functional-pdt} shows that a random input $x$ arrives at an unbiased leaf with probability at most $\epsilon$.

The bound on the tree size, which is $\min \{ p^K, 2^{O_p(A^2)} \cdot K^{O_p(A)} \}$, also follows in the same way as in the proof of Lemma \ref{lem:functional-pdt}, using a similar recursion formula whose solution is similar to the formula on Theorem \ref{thm:p-PDT}.
\end{proof}

Finally, we note that although this result is not stated in \cite{KushilevitzMansour93}, the algorithm from Lemma \ref{alg:KM} can be modified in the straightforward way to work equally well for functions $f \fptb$, with virtually the same proof.

\begin{proof}[Proof of Theorem \ref{thm:p-approximation}]
We use the algorithm from Lemma \ref{alg:KM} to find $f$'s largest Fourier coefficient in absolute value, $\hat{f}(\alpha)$. Whenever $|\hat{f}(\alpha)|\le 1-\epsilon$, the same algorithm can be used to find the second largest coefficient, $\hat{f}(\beta)$, in polynomial time (in $n$, $1/\epsilon$ and $\log (1/\delta)$). We use Lemma \ref{lem:p-functional-pdt} to construct a functional \ppdt, and replace every function-labeled leaf with the constant it is biased towards.

We again mention, as in the proof of Theorem \ref{thm:approximation}, that we do not need to calculate $\hat{f}(\alpha)$ and $\hat{f}(\beta)$ exactly, but only to within an error of, say, $\epsilon/(2pA)$, which can be guaranteed (with high probability) by the algorithm of Lemma \ref{alg:KM}.
\end{proof}

%%%%%%%%%%%%%%%%%%%%%%%%%%%%%%%%%%%

\section{Conclusions and open problems}\label{sec:open}

In this work we obtained structural results for Boolean functions over $\Z_p^n$, for prime $p$. Our results provide a more refined structure than the one given in the works of Green and Sanders \cite{GreenSanders08,GreenSanders08b}. For a certain range of parameters we also obtain improved results in the setting of the works \cite{GreenSanders08,GreenSanders08b}.

We were also able to achieve new results in the field of computational learning theory by showing that such functions can be learned with \pdt{}s as the class of hypotheses.

There are still many intriguing open problems related to the structure of Boolean functions with small spectral norm. Most of these are related to the tightness of our results (as well as to the tightness of the results of Green and Sanders \cite{GreenSanders08}). \\

We do not believe that the bound given in Equation~\eqref{eq:GS} is tight. Perhaps it is even true that one could represent $f$ as a sum of polynomially (in $A$) many characteristic functions of subspaces (note that this is not true for functions over general abelian groups. See \cite{GreenSanders08b}). Similarly, we do not believe that the bounds we obtain in Theorems~\ref{thm:PDT} and \ref{thm:p-PDT} are tight. %\footnote{Indeed, Tsang et al. obtained better bounds \cite{TsangWXZ13}, but the question that we ask is also valid for their result.}
It seems more reasonable to believe that the true bound should be $\poly(n,A)$.\\

% The results in Theorems~\ref{thm:local} and \ref{thm:p-local} are more likely close to being tight, but still, it may be the case that there is a subspace of co-dimension $O(A)$ on which the function is constant. \\

Recall that \cite{ZhangShi10,MontanaroOsborne10}  conjectured that  Boolean functions  with sparse Fourier spectrum can be computed by a \pdt{} of depth $\poly(\log \spar{f})$. Theorems~\ref{thm:sparse} and \ref{thm:p-sparse}  give an affirmative answer only for the case that $f$ also has a small spectral norm. Thus, the general case is still open. \\

Finally, Theorems~\ref{thm:approximation} and \ref{thm:p-approximation} give shallow {\ppdt}s approximating functions with small spectral norm. These results too do not seem tight. In particular, it is interesting to understand whether something better can be obtained if we assume in addition that $f$ can be computed exactly by a small \ppdt. Namely, can one output a shallow \ppdt{} approximating $f$ over the uniform distribution using polynomially many membership queries (i.e. oracle calls) to $f$, assuming that $f$ can be exactly computed by such a \ppdt{} (and has a small spectral norm).

%%%%%%%%%%%%%%%%%%%%%%%%%%%%%%%%%%%

\bibliographystyle{alpha}
\bibliography{../../bibliography}

\newcommand{\etalchar}[1]{$^{#1}$}
\begin{thebibliography}{TWXZ13}

\bibitem[ABF{\etalchar{+}}08]{ABFKP08}
M.~Alekhnovich, M.~Braverman, V.~Feldman, A.~R. Klivans, and T.~Pitassi.
\newblock The complexity of properly learning simple concept classes.
\newblock {\em {J. of Computer and System Sciences}}, 74(1):16--34, 2008.

\bibitem[BCH{\etalchar{+}}96]{BCHKS}
M.~Bellare, D.~Coppersmith, J.~H{\aa}stad, M.~A. Kiwi, and M.~Sudan.
\newblock Linearity testing in characteristic two.
\newblock {\em IEEE Transactions on Information Theory}, 42(6):1781--1795,
  1996.

\bibitem[BdW02]{BuhrmanW02}
H.~Buhrman and R.~de~Wolf.
\newblock Complexity measures and decision tree complexity: a survey.
\newblock {\em {Theoretical Computer Science}}, 288(1):21--43, 2002.

\bibitem[BLR93]{BLR93}
M.~Blum, M.~Luby, and R.~Rubinfeld.
\newblock Self­testing/correcting with applications to numerical problems.
\newblock {\em J. of Computer and System Sciences}, 47(3):549--595, 1993.

\bibitem[Fri98]{Friedgut98}
E.~Friedgut.
\newblock Boolean functions with low average sensitivity depend on few
  coordinates.
\newblock {\em Combinatorica}, 18(1):27--35, 1998.

\bibitem[GL89]{GoldreichLevin89}
O.~Goldreich and L.~A. Levin.
\newblock A hardcore predicate for all one-way functions.
\newblock In {\em Proceedings of the 21st STOC}, pages 25--32, 1989.

\bibitem[GOS{\etalchar{+}}11]{GOSSW:08}
P.~Gopalan, R.~O'Donnell, R.~Servedio, A.~Shpilka, and K.~Wimmer.
\newblock Testing fourier dimensionality and sparsity.
\newblock {\em SIAM Journal on Computing}, 40(4):1075--1100, 2011.

\bibitem[Gro97]{Grolmusz97}
V.~Grolmusz.
\newblock {On the power of circuits with gates of low $L_1$ norms}.
\newblock {\em Theoretical computer science}, 188(1):117--128, 1997.

\bibitem[GS08a]{GreenSanders08}
B.~Green and T.~Sanders.
\newblock Boolean functions with small spectral norm.
\newblock {\em GAFA}, 18:144--162, 2008.

\bibitem[GS08b]{GreenSanders08b}
B.~Green and T.~Sanders.
\newblock A quantitative version of the idempotent theorem in harmonic
  analysis.
\newblock {\em Annals. of Math.}, 168(3):1025--1054, 2008.

\bibitem[H{\aa}s01]{Hastad01}
J.~H{\aa}stad.
\newblock Some optimal inapproximability results.
\newblock {\em J. ACM}, 48(4):798--859, 2001.

\bibitem[Kal02]{Kalai02}
G.~Kalai.
\newblock {A Fourier-theoretic perspective on the Condorcet paradox and Arrow's
  theorem}.
\newblock {\em Advances in Applied Mathematics}, 29(3):412--426, 2002.

\bibitem[KKL88]{KKL88}
J.~Kahn, G.~Kalai, and N.~Linial.
\newblock {The influence of variables on Boolean functions}.
\newblock In {\em Proceedings of the 29th annual FOCS}, pages 68--80, 1988.

\bibitem[KM93]{KushilevitzMansour93}
E.~Kushilevitz and Y.~Mansour.
\newblock {Learning Decision Trees Using the Fourier Spectrum}.
\newblock {\em SIAM J. Comput.}, 22(6):1331--1348, 1993.

\bibitem[LMN93]{LMN93}
N.~Linial, Y.~Mansour, and N.~Nisan.
\newblock Constant depth circuits, {Fourier} transform and learnability.
\newblock {\em J. ACM}, 40(3):607--620, 1993.

\bibitem[Man94]{Mansour1994}
Y.~Mansour.
\newblock {Learning Boolean functions via the Fourier transform}.
\newblock In {\em Theoretical advances in neural computation and learning},
  pages 391--424, 1994.

\bibitem[MO09]{MontanaroOsborne10}
A.~Montanaro and T.~Osborne.
\newblock {On the communication complexity of XOR functions}.
\newblock {\em CoRR}, abs/0909.3392, 2009.

\bibitem[O'D12]{OdonnellBook}
R.~O'Donnell.
\newblock {Analysis of Boolean Functions}.
\newblock \url{http://www.analysisofbooleanfunctions.org/}, 2012.

\bibitem[TWXZ13]{TsangWXZ13}
H.~Y. Tsang, C.~H. Wong, N.~Xie, and S.~Zhang.
\newblock Fourier sparsity, spectral norm, and the log-rank conjecture.
\newblock {\em CoRR}, abs/1304.1245, 2013.

\bibitem[ZS10]{ZhangShi10}
Z.~Zhang and Y.~Shi.
\newblock {On the parity complexity measures of Boolean functions}.
\newblock {\em Theoretical Computer Science}, 411(26-28):2612 -- 2618, 2010.

\end{thebibliography}

\appendix

\section{Proof of Lemma~\ref{lem:simple facts}}
\label{app:missing lemma}

The proof of Lemma~\ref{lem:simple facts} relies upon the following even simpler lemma.

\begin{lemma}
\label{lem:characteristic-of-subspace}
Let $V \subseteq \Z_2^n$ be an affine subspace of co-dimension $k$, and let $\ind{V} : \Z_2^n \to \{0,1\}$ be its characteristic function. Then $\spar(\ind{V}) = 2^k$ and $\Lone{\widehat{\ind{V}}} = 1$.
\end{lemma}

\begin{proof}
Denote $V = \alpha + U$ where $U$ is a subspace of co-dimension $k$. There are $k$ vectors $\gamma_1,\ldots,\gamma_k \in \Z_2^n$ (a basis for $U^{\perp}$) and $b_1,\ldots,b_k \in \bits$ such that $\ind{V}(x) = 1$ if and only if $\chi_{\gamma_i}(x) = b_i$ for all $1 \le i \le k$. Therefore
$$
\ind{V}(x) = \prod_{i=1}^k \left( \frac{\chi_{\gamma_i}(x) + b_i}{2} \right).
$$
Using the relation $\chi_{\beta}\chi_{\gamma} = \chi_{\beta + \gamma}$, and the fact that $\mspan\{\gamma_1,\ldots,\gamma_k\} = U^{\perp}$, we get
$$
\ind{V}(x) = \sum_{\gamma \in U^{\perp}} \pm 2^{-k} \chi_{\gamma}(x).
$$
Since $|U^{\perp}| = 2^k$, both statements follow.
\end{proof}

\begin{proof}[Proof of Lemma \ref{lem:simple facts}]
Let $L$ be the set of leaves of $T$, and for every $\ell \in L$ let $b_\ell$ be its label, and $\ind{\ell} : \Z_2^n \to \{0,1\}$ be the characteristic function of the set of inputs $x$ such that computation upon $x$ arrives at the leaf $\ell$.
Since $T$ computes $f$, we may represent $f$ as:
$$
f = \sum_{\ell \in L} b_\ell \ind{\ell}(x).
$$
Now note that if $\ell$'s depth is $t$, then $\ind{\ell}$ is a characteristic function of an affine subspace of co-dimension $t$. The maximal depth of $T$ is $k$, hence for every $\ell \in L$ we have, by Lemma \ref{lem:characteristic-of-subspace}, $\spar(\ind{\ell}) \le 2^k$ and $\Lone{\widehat{\ind{\ell}}} = 1$. Finally, since $|L|=m$, we get
$$
\spar(f) \le \sum_{\ell \in L} \spar(\ind{\ell}) \le m2^k,
$$
and since $|b_\ell| = 1$, the triangle inequality implies
$$
\Lone{\hat{f}} \le \sum_{\ell \in L} \Lone{\widehat{\ind{\ell}}} \le m.
$$
\end{proof}

\section{Proof of Lemma~\ref{lem:triangle-inequality}}\label{app:triangle}

\begin{proof}
Suppose without the loss of generality (by applying a suitable rotation and reflection if needed) that $z_1=R$ is a positive real number, and that the angle is exactly $\theta \le \pi$ (i.e. $z_2 = re^{i\theta}$).

Note that $|z_1| + |z_2| = R+r$ and $z_1+z_2 = (R+r\cos(\theta))+ir\sin(\theta)$. Hence,
$$
|z_1 + z_2| = \sqrt{(R+r\cos(\theta))^2 + (r\sin(\theta))^2 } = \sqrt{R^2+r^2+2Rr\cos(\theta)}\;.
$$
It remains to be shown that
$$
R+r - \sqrt{R^2+r^2+2Rr\cos(\theta)} \ge \frac{1-\cos(\theta)}{2}r.
$$
This is equivalent to
$$
\left( R + r - \frac{1-\cos(\theta)}{2}r \right)^2 - R^2-r^2-2Rr\cos(\theta) \ge 0 .
$$
Rearranging and factoring out $r\ge0$, we get a linear function in $r$ which is non-negative on both $r=0$ and $r=R$, which implies the inequality holds for all $0 \le r \le R$.

\end{proof}

\section{Proof of Lemma~\ref{lem:overcounting}}
\label{app:overcounting}
\begin{proof}
It is enough to show that on every coset $\mygamma \in \Z_p^n/\langle \myeta \rangle$:
\begin{align}\label{eq:overcount}
3&\cdot\left(\sum_{k=0}^{p-1}{\left|\hat{f}(\mygamma+k \myeta)\right|} - \left|\sum_{k=0}^{p-1}\hat{f}(\mygamma+k\myeta)\omega^{\lambda k}\right|\right)\\
&\ge \sum_{k=0}^{p-1}{|\hat{f}(\mygamma+k \myeta)|+|\hat{f}(\mygamma+(k+1) \myeta)|- |\hat{f}(\mygamma+k \myeta)+\hat{f}(\mygamma+(k+1) \myeta) \omega^{\lambda}|}\nonumber
\end{align}
Fix a coset $\mygamma$. For $k=0,\ldots, p-1$ denote by $z_k \eqdef \hat{f}(\mygamma + k\myeta) \cdot \omega^{\lambda  k}$.
Rewriting Equation~\eqref{eq:overcount} under this notation gives
\begin{equation*}%\label{eq:overcountsim}
3 \cdot \left(\sum_{k=0}^{p-1}{|z_k\cdot \omega^{-\lambda k}|} - \left|\sum_{k=0}^{p-1}{z_k}\right|\right) \ge \sum_{k=0}^{p-1}{|z_k\cdot \omega^{-\lambda k}| + |z_{k+1}\cdot \omega^{-\lambda (k+1)}| -  |z_k\cdot \omega^{-\lambda k}+ z_{k+1}\cdot \omega^{-\lambda k}|}\;.
\end{equation*}
Since multiplying by $\omega^{-\lambda k}$ does not change the norm, this is equivalent to
\begin{equation}\label{eq:overcountsimer}
3 \cdot \left(\sum_{k=0}^{p-1}{|z_k|} - \left|\sum_{k=0}^{p-1}{z_k}\right|\right) \ge \sum_{k=0}^{p-1}{|z_k| + |z_{k+1}| -  |z_k+ z_{k+1}|}\;.
\end{equation}
We break the right hand side of Equation~\eqref{eq:overcountsimer} into 3 sums:
\begin{enumerate}
\item $\sum_{k \in \{0,2,\ldots,p-3\}}{|z_k| + |z_{k+1}| -  |z_k+ z_{k+1}|}$
\item $\sum_{k \in \{1,3,\ldots,p-2\}}{|z_k| + |z_{k+1}| -  |z_k+ z_{k+1}|}$
\item $\sum_{k \in \{p-1\}}{|z_k| + |z_{k+1}| -  |z_k+ z_{k+1}|}$
\end{enumerate}
Each sum goes over a disjoint set of pairs $(k,k+1)$. Next, we show that each sum is at most $\sum_{k=0}^{p-1}{|z_k|} - \left|\sum_{k=0}^{p-1}{z_k}\right|$, completing the proof.
We claim in general that if $A\subseteq \{0,\ldots, p-1\}$ is a subset such that $1+A = \{1+a \mod p :a\in A\}$ is disjoint of $A$, then
\[
\sum_{k\in A} {|z_k|+|z_{k+1}|-|z_k + z_{k+1}|} \le \sum_{k= 0}^{p-1} {|z_k|}- \left|\sum_{k=0}^{p-1}{z_{k}}\right|\;.
\]
Let $B := \{0,1,\ldots,p-1\} \setminus (A \cup (1+A))$, then $A \cup (1+A) \cup B$ is a disjoint union of $\{0,\ldots,p-1\}$ and we have:
\begin{align*}
\sum_{k\in A}& {|z_k|+|z_{k+1}|-|z_k + z_{k+1}|} \\
&= \sum_{k\in A} {|z_k| + |z_{k+1}|} + \sum_{k \in B}{|z_k|} - \sum_{k \in B}{|z_k|} - \sum_{k \in A}{|z_k + z_{k+1}|}\\
&= \left( \sum_{k\in A} {|z_k|} + \sum_{k \in 1+A} |z_k| + \sum_{k \in B}{|z_k|} \right) - \sum_{k \in B}{|z_k|} - \sum_{k \in A}{|z_k + z_{k+1}|}\\
&\le \sum_{k=0}^{p-1} {|z_k|} - \left|\sum_{k \in B}{z_k} + \sum_{k \in A}{ z_k + z_{k+1}}\right|\\
&= \sum_{k=0}^{p-1} {|z_k|} - \left|\sum_{k \in B}{z_k} + \sum_{k \in A}{ z_k} + \sum_{k \in 1+A} {z_k} \right|\\
&= \sum_{k=0}^{p-1} {|z_k|} - \left|\sum_{k=0}^{p-1}{z_k}\right|\;,
\end{align*}
where in the inequality we used the triangle inequality.
\end{proof}

\end{document}